\documentclass[pdflatex,sn-mathphys-num]{sn-jnl}% Math and Physical Sciences Numbered Reference Style
\usepackage{graphicx}%
\usepackage{multirow}%
\usepackage{amsmath,amssymb,amsfonts}%
\usepackage{amsthm}%
\usepackage{mathrsfs}%
\usepackage[title]{appendix}%
\usepackage{xcolor}%
\usepackage{textcomp}%
\usepackage{manyfoot}%
\usepackage{booktabs}%
\usepackage{algorithm}%
\usepackage{algorithmicx}%
\usepackage{algpseudocode}%
\usepackage{listings}%
\usepackage{xspace}
\usepackage{fancyhdr}
\usepackage{tcolorbox}
\usepackage{float}
\usepackage{hhline}         % double line in tables
\usepackage{colortbl}       % lovely colours
\usepackage{tabularray}     % fancy tables
\usepackage{dirtytalk}      % quotes

\usepackage{enumitem}       % proper alignment of bullet points even if they are changed
\usepackage{thmtools}
\usepackage{amssymb}
\usepackage{graphicx}
\usepackage{booktabs}
\usepackage{pifont}% http://ctan.org/pkg/pifont
\newcommand{\xmark}{\text{\ding{55}}}%
\usepackage{cancel}
\usepackage{mathtools}
\usepackage{cleveref}
\usepackage{tabularx}
\newcommand{\defstyle}[1]{\emph{#1}}
\newcommand{\ie}{i.\,e., }
\newcommand{\y}{\ensuremath{\checkmark}}
\newcommand{\n}{\textcolor{red}{\xmark}}

\newcommand{\BP}{\ensuremath{BP}\xspace}

\newcommand{\SC}{\ensuremath{SC}\xspace}
\newcommand{\RP}{\ensuremath{RP}\xspace}
\newcommand{\RV}{\ensuremath{RV}\xspace}

\newcommand{\SV}{\ensuremath{SV}\xspace}
\newcommand{\GETCHA}{\ensuremath{Sm}\xspace}
\newcommand{\SM}{\ensuremath{Sm}\xspace}

\newcommand{\Strength}{\mathit{strength}\xspace}

\newcommand{\margin}[1]{\ensuremath{m(#1)}\xspace}
\newcommand{\marginExt}[2]{\ensuremath{m_{#2}(#1)}\xspace}

\newcommand{\margingraph}{\ensuremath{\mathcal{M}}\xspace}

\newcommand{\margingraphRV}{\ensuremath{\margingraph^{RV}}\xspace}
\newcommand{\margingraphRVExt}[1]{\ensuremath{\margingraph^{RV}(#1)}\xspace}
\newcommand{\Pref}{\ensuremath{\mathbf{P}}\xspace}
\newcommand{\SCF}{\ensuremath{F}\xspace}
\newcommand{\pref}{\ensuremath{\succ}\xspace}

\theoremstyle{thmstyleone}%
\newtheorem{theorem}{Theorem}[section]% meant for sectionwise numbers
\newtheorem{proposition}[theorem]{Proposition}% 
\usepackage{thm-restate}
\newtheorem{example}{Example}

\newtheorem*{lemma*}{Lemma}

\newtheorem{observation}[theorem]{Observation}
\newtheorem*{observation*}{Observation}

\theoremstyle{thmstyletwo}%

\theoremstyle{thmstylethree}%
\newtheorem{definition}[theorem]{Definition}

\usepackage{thm-restate}
\newcommand{\restatemarker}{\textup{($\star$)} \ignorespaces}

\begin{document}

\title[Article Title]{The River Method}


\author*[1]{\fnm{Michelle} \sur{Döring}}\email{michelle.doering@hpi.de}

\author[2]{\fnm{Markus} \sur{Brill}}\email{markus.brill@warwick.ac.uk}

\author[3]{\fnm{Jobst} \sur{Heitzig}}\email{jobst.heitzig@pik-potsdam.de}

\affil*[1]{\orgdiv{Algorithm Engineering}, \orgname{Hasso Plattner Institut}, \orgaddress{
    %\street{Prof.-Dr.-Helmert-Straße 2-3}, 
    \city{Potsdam}, %\postcode{14482}, \state{Brandenburg}, 
    \country{Germany}}}

\affil[2]{\orgdiv{Department of Computer Science}, \orgname{University of Warwick}, 
    \city{Coventry}, %\postcode{CV4 7AL}, 
    \country{United Kingdom}}

\affil[3]{\orgdiv{FutureLab on Game Theory \& Networks of Interacting Agents, Complexity Science Department}, \orgname{Potsdam Institute for Climate Impact Research}, \orgaddress{
    %\street{Telegrafenberg A 31}, 
    \city{Potsdam}, %\postcode{14473}, \state{Brandenburg}, 
    \country{Germany}}}


\abstract{We introduce \textit{River}, a novel Condorcet-consistent voting method that is based on pairwise majority margins and can be seen as a simplified variation of Tideman's \textit{Ranked Pairs} method. River is simple to explain, simple to compute even ``by hand,'' and gives rise to an easy-to-interpret certificate in the form of a directed tree. Like Ranked Pairs and Schulze's \textit{Beat Path} method, River is a refinement of the \textit{Split Cycle} method and shares with those many desirable properties, including independence of clones. Unlike the other three methods, River satisfies a strong form of resistance to agenda-manipulation that is known as \textit{independence of Pareto-dominated alternatives.}}




\maketitle

\section{Introduction}
The task of making a collective decision on the basis of individual rankings is fundamental to social choice theory \cite{arrow_handbook_2010,brandt_handbook_2016} and has a wide range of applications in artificial intelligence \cite{furnkranz2003pairwise,askell2021general, kopf2024openassistant,mishra2023ai}. For instance, K\"opf et. al. \cite{kopf2024openassistant} have recently employed a rank-aggregation approach to align large language models (LLMs) with human preferences. 

We use the framework of \textit{voting theory} \cite{zwicker2016introduction} and interpret rankings as preference orders given by a set of \textit{voters} over a set of \textit{alternatives}. A key principle in collective decision making is \textit{majority rule}---the idea that the decision should follow what is seen as ``the will of the majority'' \cite{May52a}.
In the simple case of only two alternatives, majority rule is unambiguous and chooses the alternative that is ranked first by more than half of the voters.
This principle is used by many common decision methods which are based on \textit{pairwise comparisons} between the alternatives.

In real-world elections, a \textit{Condorcet winner} often emerges---an alternative that defeats all others in pairwise comparison. Often, this alternative is considered to be the most suitable choice for the election winner.
In the absence of a Condorcet winner, each alternative faces at least one majority defeat. 
To then decide on a winner, one can encode the pairwise majority comparisons as a graph, with each alternative represented by a vertex, and an edge $(y,x)$ indicating that $y$ defeats $x$. This forms a \textit{tournament graph}, and a selection process based on this is known as a \textit{tournament solution} \cite{brandt_tournament_2016,Lasl97a}.
Here, majority rule is embodied by \textit{``beatpaths''}:
For any chosen winner $x$, if there is a defeat $y\to x$, there should be a path of defeats leading from $x$ back to $y$:  $x\to z_1\to\cdots\to z_k=y$.
These alternatives form the \textit{Smith set} (a.k.a.~\textit{GETCHA} or \textit{top cycle}) \cite{schwartz_logic_1986,Good71a,Smit73a}. 

While tournament theory extends this concept by studying several refinements of the Smith set based solely on pairwise majority defeats, another strand of research focuses on the strength of these majority defeats, measured by the numerical majority margin (the number of voters ranking $x$ over $y$ minus the number of voters with opposite rankings) \cite{fischer_weighted_2016}.
The tournament graph with edges weighted by this margin is called the \textit{margin graph}.
% This leads to an integer-valued, antisymmetric weighted graph, called the \textit{margin graph}. % and used to filter the alternatives further or even completely decide the winner. 

A natural extension of the Smith set from the tournament graph to the margin graph is that for any defeat $y\to x$ of a winner $x$, there should be a ``rebutting'' beatpath from $x$ to $y$, with each defeat in the path being at least as strong as $y\to x$. 
This allows defending the choice of $x$ against claims of the form ``a majority ranks $y$ over $x$'' by pointing to a sequence of equally strong (or stronger) claims of the same form leading back to $y$.
% This allows defending the choice of $x$ against claims like ``a majority prefers $y$ over $x$'' by pointing out that if such a claim were valid, there would be a sequence of equally strong (or stronger) claims of the same type leading back to $y$. 
An alternative fulfilling this property is called \textit{immune}.
Since at least one immune alternative always exists in every election, this approach can be seen as a natural way to operationalize the concept of majority rule \cite{heitzig2002social,dung1995acceptability}. 
The corresponding voting method, choosing all immune alternatives, is called \textit{Split Cycle} \cite{holliday2020split}. 
This method has many appealing properties but it often chooses multiple alternatives as the  winner.
There exist several popular methods that always choose a subset of immune alternatives and are therefore refinements of Split Cycle, most notably \textit{Ranked Pairs} \cite{tideman1987independence} and \textit{Beat Path} \cite{schulze2011new}. Both typically select a single winner, except in rare cases of ties, and each satisfies a distinct set of desirable properties.
The same is true for \textit{Stable Voting} \cite{holliday2023stable}, the most recently introduced refinement of Split Cycle.

All four methods---Split Cycle, Ranked Pairs, Beat Path, and Stable Voting---suffer from a weakness related to agenda manipulation: introducing a new alternative $z$ might alter the winning set even if $z$ is \textit{Pareto-dominated} by an existing alternative $y$ (i.e., all voters rank $y$ over $z$). 
Formally, these methods violate \textit{independence of Pareto-dominated alternatives (IPDA)}. 
This property goes beyond the standard notion of  Pareto efficiency (stating that Pareto-dominated alternatives should not be selected) and requires that Pareto-dominated alternatives should not influence the outcome at all. 
Without IPDA, it becomes possible to propose additional, Pareto-dominated alternatives in order to gain an advantage or disturb the voting process. 
As a result, IPDA is often considered a desirable property in social choice theory \cite{fishburn_theory_1973,Rich78c,Chin96a,ChSh98a,GLS19a,Oztu20a,BrPe22a}. Recently, \cite{GrCB23a} explored IPDA in the context of decision making under moral uncertainty.

\subsection*{Our contribution}
In this paper, we introduce the novel\footnote{The method was first described by Jobst Heitzig in a 2004 mailing list post \cite{river2004mailinglist}.} social choice function \textit{River}, a resolute refinement of Split Cycle that satisfies independence of Pareto-dominated alternatives (IPDA)—a property not shared by any other Split Cycle refinement.
River can be viewed as a simpler variation of Ranked Pairs: while both methods build an acyclic subgraph of the margin graph rooted at the winner, River always constructs a spanning tree, in contrast to the typically denser graphs produced by Ranked Pairs. This structural simplicity improves transparency and interpretability: the tree acts as a ``rebutting diagram,'' with a unique majority path from the winner to every alternative, certifying the winner’s immunity to majority-based challenges.

We give preliminary definitions of social choice functions, as well as the definitions of Split Cycle, Ranked Pairs, Beat Path, and Stable Voting in \Cref{sec:prelims}. Then, in \Cref{sec:river}, we give the formal definition of River and observe some basic properties. In \Cref{sec:Properties}, we show that River satisfies standard axioms, including independence of smith dominated alternatives and, most importantly, independence of pareto dominated alternatives.
Proofs of results marked $(\star)$ can be found in the appendix.

%%%%%%%%%%%%%%%%%%%%%%%%%%%%%%%%%%%%%%%%%%%%%%%%%%%%%%%%%%
% \subsection{Contribution}
%=================== \markus{do we want a contribution subsection? If yes, see commented out text}
% We provide a formal definition and thorough analysis of the proposed River (\RV).
% We compare \RV to the state-of-the-art social choice functions Ranked Pairs, Beat Path, Split Cycle, and Stable Voting.
% To that end, we compare their respective properties and extend the table of \cite{holliday2020split} by \RV. %and add IPDA
% In \Cref{sec:Implementation_and_Experiments} we provide computational evidence showcasing that the winning sets of those five voting rules coincide in almost ``all cases''. This emphasizes the main benefit of River:  While it gives almost every time the same outcome as the other more complex voting rules, it is in comparison easy to compute, explain, and defend a River winner. And it still has the many nice properties we want.

% For a more applied context, we provide an algorithm to compute the River+ diagram. With it one can justify the results of the decision process to the voters in an understandable, yet exhaustive way.
% Since aggregating full pairwise preferences is in most context involving humane voters too pricey, we propose interactive to semi-interactive variants of River to be used in small to medium sized electorates.
%%%%%%%%%%%%%%%%%%%%%%%%%%%%%%%%%%%%%%%%%%%%%%%%%%%%%%%%%%

\section{Preliminaries}
\label{sec:prelims}
Let $[k] \coloneqq \{1,\dots,k\}$.
    We consider elections with a set $N$ of $n\ge 1$ \defstyle{voters} with preferences over a set $A$ of $m\ge 2$ \defstyle{alternatives}.
    The preferences of a voter\footnote{In this paper, we do not consider strategic misrepresentation of preferences, which would require distinguishing between a voter's actual preferences and the ranking they choose on their ballot.} $i\in N$ are given as a strict ranking, e.g., a linear order ${\pref_i}$ over~$A$ and denote by $x\pref_i y$ that voter $i$ ranks alternative $x$ above $y$.
    % A \defstyle{(preference) profile} $\Pref={(\pref_1,\dots,\pref_n)}\in\mathcal{L}(A)^n$ is a list that contains the preferences of $n$ voters.
    A \defstyle{(preference) profile} $\Pref={(\pref_1,\dots,\pref_n)}$ is a list containing the preferences of all voters.
    % The set of all preference profiles over a fixed set $A$ is given by $\PrefProf = \bigcup_{n=1}^\infty {\mathcal L}(A)^{n}$.
    % Given a preference profile $\Pref$, 
    We  may denote its corresponding set of voters and alternatives by $N(\Pref)$ and $A(\Pref)$. %, respectively.
    For an alternative $x \in A$, we let $\Pref_{-x}$ denote the restriction of $\Pref$ to $A \setminus \{x\}$.
    The \defstyle{majority margin} of $x$ over $y$ according to \Pref is
    \[\marginExt{x,y}{\Pref}= \left\lvert\{i\in N \colon x\pref_i y\}\right\rvert - \left\lvert\{i\in N \colon y\pref_i x\}\right\rvert. \]
    When $\marginExt{x,y}{\Pref}>0$, we say $x$ \defstyle{defeats} $y$, denoted \mbox{$x\pref_\Pref y$}, and when $\marginExt{x,y}{\Pref}\ge 0$, we say $x$ \defstyle{weakly defeats} $y$, denoted $x\succeq_\Pref y$.
        % Whenever $\marginExt{x,y}{\Pref}>0$, we say $x$ \defstyle{defeats} $y$, and we say $x$ \defstyle{weakly defeats} $y$ if $\marginExt{x,y}{\Pref}\ge 0$.
    The \defstyle{majority graph} of a profile~$\Pref$ has vertex set $A$ and 
    an edge from $x$ to $y$ whenever $x$ weakly defeats $y$. 
    % edge set $\{(x,y)\in A\times A \colon  x\succeq_\Pref y\}$.
    % $\margingraph_\Pref=(A,\{(x,y)\in A\times A \colon  x\succeq_\Pref y\})$.
    We refer to the majority graph's edges as \textit{majority edges} and to its cycles as \textit{majority cycles}. 
    The \emph{margin graph} $\margingraph_\Pref$ of $\Pref$ is the weighted version of the majority graph where each majority edge $(x,y)$ has weight $\marginExt{x,y}{\Pref}\ge 0$. If it is clear from the context, we drop the index $\Pref$ from $x\pref y$, $\margin{x,y}$ and $\margingraph$.
    A \defstyle{majority path} in \margingraph\ is a sequence $p=(x_1,\dots,x_\ell)$ of distinct alternatives such that for $i\in[\ell-1]$, $\margin{x_i,x_{i+1}}>0$. The \defstyle{strength} of such a path is its lowest margin,
    \[\Strength(p)=\min\{\margin{x_i,x_{i+1}}\colon i\in[\ell-1]\}.\]
    Analogously, \defstyle{majority cycles} in \margingraph are closed paths ($x_\ell=x_1$) such that $\margin{x_i,x_{i+1}}>0$ for all continuous pairs, and their strength is defined as the lowest margin on the cycle.
    
    A \defstyle{Condorcet winner} $x$ is an alternative which defeats all other alternatives, \ie $x\pref y$ for all $y\in A\setminus\{x\}$. Such an alternative does not always exist.
    % Conversely, a \defstyle{Condorcet loser} is defeated by all alternatives.
    A natural extension is the notion of dominant sets. A nonempty set $X\subseteq A$ is \defstyle{dominant}, if 
    $x\pref y$ for all $x\in X$ and $y\in A\setminus X$.
    There always exists at least one dominant set, e.\,g., the whole set of alternatives $A$.

\subsection{Immunity against Majority Complaints}
\label{sec:immu}

    In the absence of a Condorcet winner, every alternative suffers at least one majority defeat, potentially eliciting a 'majority complaint'.
    Certain alternatives can be defended against such complaints by showing the existence of majority paths from the defended alternative to each alternative defeating it.
    In the \textit{unweighted} setting, the \textit{Smith set} consists precisely of the set of alternatives than can be defended in this way. Formally, the Smith set $\GETCHA(\Pref)\subseteq A$ of a preference profile $\Pref$ is defined as the unique inclusion-minimal dominant subset of alternatives.
    % \begin{definition}%[Smith set]
    %     The \defstyle{Smith set} $\GETCHA(\Pref)\subseteq A$ of a preference profile $\Pref$ is the unique inclusion-minimal dominant subset.  
    % \end{definition}
    Such resistance against majority complaints can be generalized to the weighted setting in a straightforward way by taking the strength of defeats into account \cite{heitzig2002social,holliday2020split}.
    % An alternative $x$ is called \defstyle{immune (against majority complaints)} if for every alternative $y$ defeating $x$, there is a  majority path from $x$ to $y$ that is \textit{at least as strong} as the margin of $y$ over $x$. % to defend $x$ over $y$.
    \begin{definition} %[Immunity] 
    \label{def:Immunity}
    Given a preference profile $\Pref$, an alternative $x\in A$ is called \defstyle{immune}, if for every $y\in A\setminus\{x\}$ with $\margin{y,x}>0$, there exists a majority path $p$ in \margingraph from $x$ to $y$ with $\Strength(p)\geq \margin{y,x}$.
    \end{definition}

\subsection{Social Choice Functions}
\label{sec:scf}
 
 A \textit{social choice function (SCF)} $\SCF$ maps a preference profile $\Pref$ to a non-empty set $F(\Pref)\subseteq A$ of \textit{winning} alternatives. Given two social choice functions $F$ and $G$, we call $F$ a \defstyle{refinement} of $G$ if $F(\Pref)\subseteq G(\Pref)$ for all  profiles $\Pref$.
    % We focus on two important classes of SCF's: \defstyle{majoritarian} SCF's and \defstyle{margin-based} SCF's.
    % A majoritarian SCF solely considers the pairwise majority relation, while a margin-based SCF additionally takes the margins of the pairwise majority relations into account. That is, for a margin-based SCF \SCF holds for all preference profiles $\Pref, \Pref'$, if $\margingraph_\Pref=\margingraph_{\Pref'}$ then $\SCF(\Pref)=\SCF(\Pref')$. 

    The social choice function Split Cycle (\SC) \cite{holliday2020split} selects all immune alternatives as winners. 
    Formally, for each majority cycle $C$ in \margingraph, an edge in $C$ with margin $\Strength(C)$ is called a \textit{splitting edge}.
    Splitting edges are, therefore, exactly the edges with the lowest margin within the cycle.
    The \defstyle{Split Cycle diagram}, denoted $\margingraph^{\SC}$, is the subgraph of \margingraph obtained by removing all splitting edges.
    As a result, $\margingraph^{\SC}$ contains no majority cycles.
    The winners under Split Cycle are the alternatives with no incoming edge in $\margingraph^{\SC}$. 
        % If $\margin{x,y}>0$ and $(x,y)$ is not a splitting edge of a majority cycle, then $x$ ``beats'' $y$ according to Split Cycle. We denote by $\margingraph^{\SC}$ the subgraph of $\margingraph$ containing only those edges.
        % The Split Cycle winners are all unbeaten alternatives. 
    % \smallskip
    %
    In this paper, we consider three refinements of Split Cycle: 
    \textit{Ranked Pairs}  \cite{tideman1987independence},
    \textit{Beat Path}  \cite{schulze2011new}, and 
    \textit{Stable Voting} \cite{holliday2023stable}.
    
    \paragraph{Ranked Pairs (\RP)}
        Starting with an empty graph on vertices $A$, add edges from \margingraph one at a time in order of decreasing margin, skipping any edge that would create a cycle.
        % In order of decreasing margin, add the edges of the margin graph \margingraph\ to an initially empty subgraph $\margingraph^{\RP}$ of \margingraph, skipping edges that create cycles.
        In case of ties in the margins, break them using a tiebreaker. % (a linear ordering over the edges).
        The unique alternative without incoming edges in the resulting \defstyle{Ranked Pairs diagram} $\margingraph^{\RP}$ is the Ranked Pairs winner.
            
    \paragraph{Beat Path (\BP)} %A majority path from $x$ to $y$ is a sequence $p_1,\dots,p_n$ of distinct alternatives with $p_1=x$ and $p_n=y$ such that for $1\leq k\leq n-1$, $p_k$ is majority-preferred to $p_{k+1}$, \ie $p_i\pref_\Pref p_{i+1}$. The \defstyle{strength} of a path is its lowest margin.
        Compute the strongest majority path between any pair $(x,y)$.
        If the strength of the strongest path from $x$ to $y$ is at least as 
        great as the strength of the strongest path from $y$ to~$x$, then $x$ is said to ``beat'' $y$ under Beat Path.
        The Beat Path winners are all unbeaten alternatives. 
        
    \paragraph{Stable Voting (\SV)}
        The winners are defined recursively.
        If there is only one alternative in $A$, then that alternative wins.
        % If there are at least two alternatives, order the edges $(x,y)$ of the margin graph by decreasing margin such that $x$ is undefeated according to Split Cycle.
        If there are at least two alternatives, consider all edges $(x,y)$ of \margingraph such that $x$ is immune (Split Cycle winner). 
        Order these edges by decreasing margin. In case of ties in the margins, break them using a tiebreaker.
        In that order, check for each edge $(x,y)$ whether $x$ is a Stable Voting winner in the election without $y$, and if so, declare $x$ as the Stable Voting winner.  

    \smallskip
    Since Ranked Pairs, Beat Path and Stable Voting are refinements of Split Cycle, they always select immune alternatives.
    Additionally, Split Cycle and Ranked Pairs offer a certificate for the immunity of the winner via ``rebutting'' paths in the graphs $\margingraph^{\SC}$ and $\margingraph^{\RP}$.
    \Cref{fig:example_SC_functions} illustrates the social choice functions on an example profile with six alternatives.
    \begin{figure}[t]
        \centering
        \includegraphics[width=\textwidth]{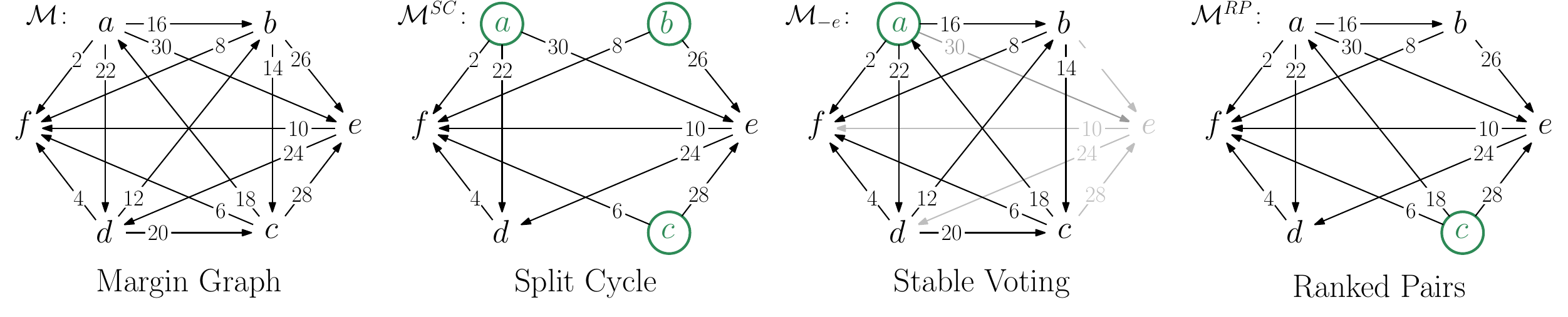}
        \caption{An election with 6 alternatives showing the behavior of the different social choice functions. From left to right are the margin graph $\margingraph$, the Split Cycle graph $\margingraph^{\SC}$ with winning set $\{a,b,c\}$, the Stable Voting winner $a$ with the deciding election without $e$, and the Ranked Pairs graph $\margingraph^{\RP}$ with winner $c$. %ss, and a table of all beatpaths with Beat Path winner $b$.
        }
        \label{fig:example_SC_functions}
    \end{figure}

    \subsection{Tiebreaking and Uniquely Weighted Profiles}
    \label{sec:ties}
    Computing winners for Ranked Pairs and Stable Voting involves ordering majority edges by their margin. With sufficiently many voters, \textit{ties} between margins (\ie two edges with the same margin) are rare, and we often restrict our attention to profiles where ties do not occur.
    \begin{definition}%[Uniquely-weighted profiles]
        A preference profile \Pref is \defstyle{uniquely weighted} if for all alternatives $v,w,x,y\in A_\Pref$ with $v\neq x$ or $w\neq y$, we have $m_{\Pref}(x,y)\neq m_{\Pref}(v,w)$ and $m_{\Pref}(x,y)\neq0$.
    \end{definition}
    An SCF is called \textit{resolute} if it always outputs a single alternative. 
    Split Cycle is not resolute, even for uniquely weighted profiles (see e.\,g., \Cref{fig:example_SC_functions}), while Ranked Pairs, Beat Path, and Stable Voting are resolute on these profiles.
    For preference profiles that are not uniquely weighted, \textit{tiebreakers} are required to compute the winning set.
    \begin{definition}[\cite{zavist1989complete,brill_price_2012}]
        A \emph{tiebreaker} $\tau = (e_1, \dots, e_{\lvert E\rvert})$ is a descending linear ordering of all edges $E$ in the margin graph of a preference profile by decreasing margin, \ie $m_{\Pref}(e_i) \geq m_{\Pref}(e_j)$ for all $1 \leq i < j \leq \lvert E\rvert$.
    \end{definition}    
        % A \textit{tiebreaker} in this sense is defined as a \textit{descending linear ordering} $\lino = (e_1, \dots, e_{\lvert E\rvert})$ of the margin edges $E(\mgg)$ by decreasing margin, \ie $\margin{e_i} \geq \margin{e_j}$ for all $1 \leq i < j \leq \lvert E\rvert$ \cite{zavist1989complete,BrFi12a}.
    This linear order can either be specified directly or derived using a \textit{tiebreaker function}, which takes a preference profile \Pref as input and returns a tiebreaker for \Pref.
    % \begin{definition}
    %     A \emph{tiebreaker function} takes a preference profile \Pref as input and returns a linear order over all edges in the margin graph of \Pref.
    % \end{definition}
    In the literature, the notion for \textit{tiebreaker} and \textit{tiebreaker function} are often used interchangeably.
    
    One common tiebreaker function sorts the edges $(x,y)$ \textit{lexicographically} based on a fixed order of the alternatives:
        first, by the source alternative $x$, and then among edges with the same source, by the target alternative $y$.
    Another example is to sort the edges uniformly at random for each new profile.
    In this paper, we focus on non-random tiebreaker functions.
    
    % One such condition is \textit{consistency}, which requires that the ordering of pairs does not change when alternatives are added or removed.        
    Tideman originally defined Ranked Pairs using a fixed, but arbitrary tiebreaker (function).
    This means that when analyzing an SCF under changes to the election\,--\,such as adding or removing alternatives\,--\, the tiebreaker remains \textit{consistent}. 
    \begin{definition}
        A tiebreaker function is \emph{consistent}, if an edge $(a,b)$ precedes $(c,d)$ in the tiebreaker on $\Pref$ if and only if it $(a,b)$ precedes $(c,d)$ in the tiebreaker for $\Pref_{-x}$, for any $x\notin\{a,b,c,d\}$.
    \end{definition}
    When discussing properties of SCFs, we specify whether the analysis assumes uniquely weighted profiles (where no tiebreaker is needed) or general profiles, in which case we may specify conditions on the tiebreaker (function).
    
    % \subsubsection*{Parallel Universe Tiebreaking}
    % Parallel Universe Tiebreaking (PUT) extends an SCF $F$ by considering all possible tiebreakers instead of a specific one.
    % The resulting SCF, \mbox{$F$-$PUT$}, takes as input a preference profile \Pref and returns all alternatives that could win according to $F$ under any valid tiebreaker. 
    % Formally, if $\mathbb{O}$ is the set of all descending linear orderings, then 
    % $$F\text{-}PUT(\Pref) := \bigcup_{\tau\in\mathbb{O}} F(\Pref,\tau).$$
    % Note that while the winning set of Split Cycle and its refinements with a given tiebreaker are each computable in polynomial time \cite{holliday2020split}, determining the winners of the PUT variant of Ranked Pairs (\mbox{$\RP$-$PUT$}) is known to be \NP-complete \cite{brill2012price}.

\section{The River Method} \label{sec:river}
    Aiming to satisfy as many axiomatic properties as possible can lead to definitions of rather complex social choice functions that are hard to analyse and non-trivial to execute.   
    % In order to design social choice functions fulfilling an increasing set of desirable properties, most recently introduced methods tended to also inhibit increased complexity.
    %While Split Cycle and Ranked Pairs are comparably easy to understand and fairly fast to compute, the computation of Beat Path and particularly Stable Voting is rather involved. 
    % It proves difficult to spot possible winners following solely the definition of the social choice function. 
    For example, \textit{Stable Voting} requires a recursive computation of Stable Voting and Split Cycle winners on smaller instances, reducing the profile one alternative at a time. 
    %A direct answer is guaranteed only for two alternatives.
    %    \textcolor{gray}{\RP - increasing number of edges locked in, aggravates check for cycles}
        
    %     In addition, tackling majority complaints becomes increasingly challenging. 
    % %    Moreover, with increasing complexity of the process it gets difficult to tackle majority complains about the chosen winner.
    %     Whenever the preferences admit no Condorcet winner, no matter which alternative is chosen as the winner, there will be at least one other alternative which is preferred by a majority.
    %     Towards such a complaint has to be given an argument, based on the given preferences, as to why the winner is actually preferable. As all the refinements of Split Cycle choose immune alternative, such a counterargument must exist.
    %     But more complex choice methods lose necessary transparency of the choice process. % and conciseness regarding the reason for choosing the winner over each of the other alternatives.
    Moreover, even if a social choice function is guaranteed to only select immune alternatives, the ``rebutting'' paths that witness the immunity are sometimes nontrivial to find. 
    For example, Ranked Pairs produces a typically very large, acyclic subgraph of the margin graph in which it is nontrivial to find rebutting paths without the help of a computer. In the appendix, \Cref{sec:app_baseballexample}, we present an example with 14 alternatives.
    
    The main motivation behind \textit{River} is simplicity. The method is \textit{(i)} simple to explain (see below for the procedural definition), \textit{(ii)} simple to compute (the winner can easily be calculated ``by hand''), and 
    \textit{(iii)} gives rise to unique rebutting paths that are easy to spot in the resulting diagram. %\michelle{the diagram is a tree, which promises mathematical advantages}

    % \subsection{Definition}
    Operating similarly to Ranked Pairs, River ``splits'' majority cycles by starting with an empty graph and iteratively adding majority edges in order of decreasing margin while maintaining acyclicity. In contrast to Ranked Pairs, for which acyclicity is the only criterion, River also avoids adding majority edges towards an already defeated alternative.
    This results in a tree of majority edges with the winner as the root.
    % The result is a method that effectively resolves majority cycles while maintaining transparency and simplicity in its decision process.
    \begin{definition} \defstyle{River} (\RV) operates as follows, given a preference profile \Pref:
    \begin{enumerate}
        \item Order the majority edges $E(\margingraph)$ by decreasing margin (using a tiebreaker if necessary).
        %\footnote{If there exist pairs of vertices with $m_\Pref(x,y)=m_\Pref(y,x)=0$, use the tiebreaker to add one of these to edges to the order.}
        \item Initialize an empty graph on the set of all alternatives $A$.
        \item Process the edges of \margingraph following the order defined in Step 1 and add an edge if this addition neither creates
        \begin{enumerate}
            \item[(Cy)] a cycle, nor
            \item[(Br)] a branching (two in-edges for an alternative). \label{RMdiffersRP}
        \end{enumerate}
        \item The River winner $\RV(\Pref)$ is the unique source in the resulting \defstyle{River diagram} \margingraphRV.
    \end{enumerate}
    \end{definition}
    \begin{figure}[t]
        \centering
        \includegraphics[width=0.7\columnwidth]{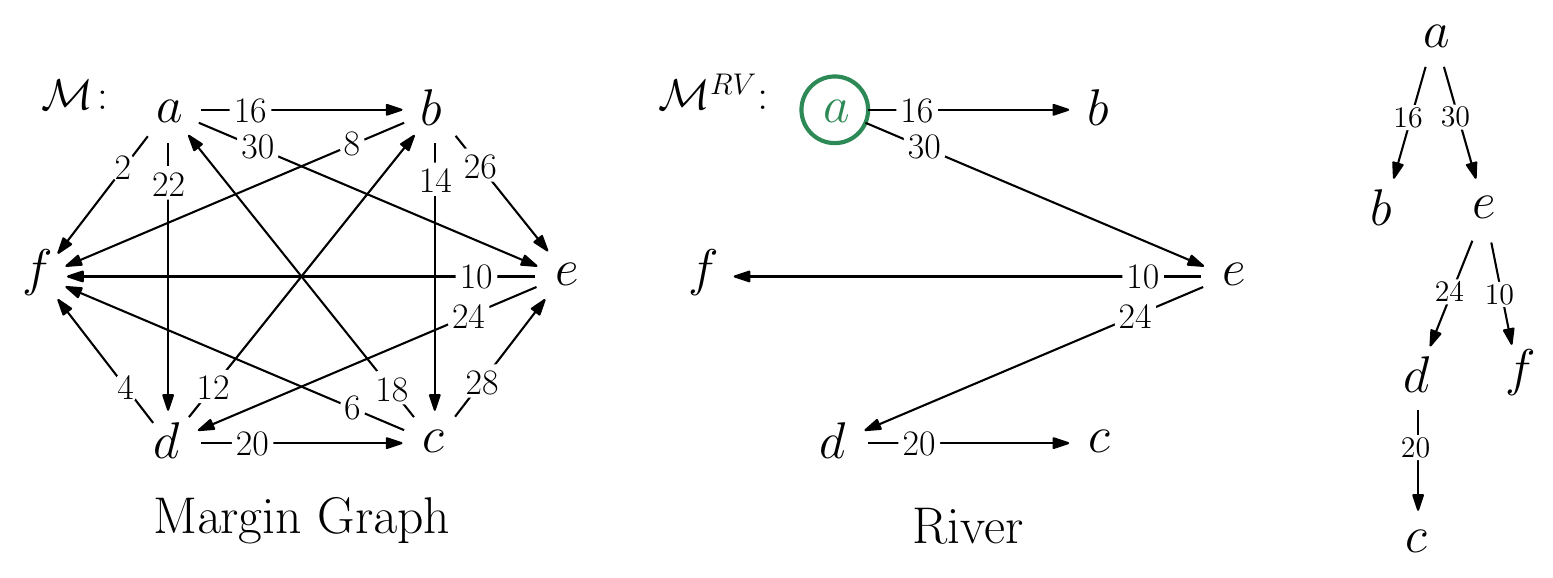}
        \caption{
        The River diagram \margingraphRV\ for the margin graph of \Cref{fig:example_SC_functions} with $\RV(\Pref)=\{a\}$.
        On the right, it is arranged as a tree with the root $a$ at the top.
        % The margin graph \margingraph\ and River diagram \margingraphRV\ for the example of \Cref{fig:example_SC_functions} with $\RV(\Pref)=\{a\}$.
        }
        \label{fig:example_definition}
    \end{figure}
    We may refer to the \defstyle{cycle condition} by $(Cy)$ and the \defstyle{branching condition} by $(Br)$.
    Note that River differs from Ranked Pairs only in $(Br)$, which ensures \margingraphRV to be a tree.
    The consequences of this subtle change will be thoroughly discussed in the rest of the paper. 
    It is easy to check that River winners are immune:
    \begin{proposition}
        River is a refinement of Split Cycle, \ie $x\in \RV(\Pref)$ implies $x\in \SC(\Pref)$.
    \end{proposition}
    \begin{proof}
        Let \Pref be a preference profile. We show the contraposition of the claim, \ie $x\notin\SC(\Pref)$ implies $x\notin\RV(\Pref)$. If $x\notin\SC(\Pref)$, there is an alternative $y\in A\setminus\{x\}$ with $(y,x)\in E(\margingraph^{\SC})$. This means the edge $(y,x)$ is not a splitting edge of any cycle in \margingraph.
        If $(y,x)\in E(\margingraphRV)$, $x\notin\RV(\Pref)$ follows. So, assume $(y,x)\notin E(\margingraphRV)$.
        If the edge was not added because of $(Br)$, there is another edge $(z,x)\in E(\margingraphRV)$ and $x\notin\RV(\Pref)$.
        If the edge was not added because of $(Cy)$, it would have closed a cycle with edges of margin at least $\margin{y,x}$, which is a contradiction to $(y,x)$ not being a splitting edge. 
        % Therefore, $x\notin\RV(\Pref)$.    
    \end{proof}
    \Cref{fig:example_definition} shows the River diagram for the margin graph from \Cref{fig:example_SC_functions}, with $\RV(\Pref)\neq\RP(\Pref)$, containing notably fewer edges than the Ranked Pairs diagram (this is even more pronounced in the larger example in appendix, \Cref{sec:app_baseballexample}).
    In fact, the \margingraphRV always contains exactly $m-1$ edges, whereas the $\margingraph^{\RP}$ may contain all $\genfrac(){0pt}{2}{m}{2}$ edges. 
    Note also that River, Ranked Pairs and Split Cycle clearly justify their choice through the respective diagrams, unlike Stable Voting and Beat Path.
    % Unlike Stable Voting and Beat Path, the three other social choice functions---River, Ranked Pairs and Split Cycle---offer a clear rationale for their choice through the corresponding diagrams.
    
    % By its definition,
    River, like Ranked Pairs and Stable Voting, requires a tiebreaker to be resolute for non-uniquely weighted preference profiles, and is resolute for uniquely weighted profiles.   
    % River shares the need for a tiebreaker in non-uniquely weighted preference profiles with Ranked Pairs and Stable Voting. Equipped with a tiebreaker % that yields a total order over the edges, 
    % the method is resolute. Obviously, River is also resolute for uniquely weighted profiles.
    \begin{proposition} \label{prop:Resolute}
        River is resolute.
    \end{proposition}
    \begin{proof}
        Using the tiebreaker, the ordering of the edges in Step~1 is strict and Step~3 is well-defined. The River diagram is then an acyclic, connected graph in which each vertex has at most one incoming edge, \ie a tree with a distinct root, and $\lvert\RV(\Pref)\rvert=1$.
    \end{proof}

    % Recently, it was shown in \cite{riverPUT2024melanowski} that the parallel-universe tiebreaking variant of River (\mbox{\RV-$PUT$}) can be computed in polynomial time\,--\,despite the equivalent problem being \NP-complete for Ranked Pairs.
    % The algorithm and runtime proof heavily depend on the tree structure of the River diagram.
    % \begin{theorem}[\cite{riverPUT2024melanowski}]
    %     The winning set of \mbox{\RV-$PUT$} can be computed in polynomial time.
    % \end{theorem}

\section{Axiomatic Properties} \label{sec:Properties}
\label{subsec:BasicAxioms}

\definecolor{DarkLine}{RGB}{156,156,156}    %dark line  
\definecolor{LightLine}{rgb}{1,1,1}         %light line 
\definecolor{DarkRow}{RGB}{222, 222, 222}   %dark row 
\definecolor{LightRow}{RGB}{242, 242, 242}  %light row  
\definecolor{White}{rgb}{1,1,1}
%----------------------------------
\definecolor{DarkRowRiver}{RGB}{212,212,212}        % DARKER dark row   rgb()
\definecolor{LightRowRiver}{RGB}{232,232,232}          % DARKER light row
\definecolor{Grayy}{RGB}{242, 242, 242}
\newcommand{\svipda}{(\y*)}
\newcommand{\RM}{RV}
    \begin{table}[t]
    \centering
    \resizebox{0.7\textwidth}{!}{%
    \footnotesize
    \begin{tabular}{lccccc}
    \rowcolor{White}
    & \SC & \SV & \RP & \BP & \cellcolor{Grayy}\RM \\
    
    \rowcolor{DarkRow}
    Anonymity & \y & \hphantom{*}\y* & \hphantom{*}\y* & \y & \cellcolor{DarkRowRiver}\hphantom{*}\y* \\
    
    \rowcolor{LightRow}
    Neutrality & \y & \hphantom{*}\y* & \hphantom{*}\y* & \y & \cellcolor{LightRowRiver}\hphantom{*}\y* \\
    
    \rowcolor{DarkRow}
    Monotonicity & \y & \n & \y & \y & \cellcolor{LightRowRiver}\y \\
    
    \rowcolor{LightRow}
    Condorcet Winner & \y & \y & \y & \y & \cellcolor{DarkRowRiver}\y \\
    
    \rowcolor{DarkRow}
    Condorcet Loser & \y & \y & \y & \y & \cellcolor{LightRowRiver}\y \\
    
    \rowcolor{LightRow}
    Smith criterion & \y & \y & \y & \y & \cellcolor{DarkRowRiver}\y \\
    
    \rowcolor{DarkRow}
    Pareto efficiency & \y & \y & \y & \y & \cellcolor{LightRowRiver}\y \\
    
    \rowcolor{LightRow}
    ISDA & \y & $\hphantom{\circ}\y^\circ$ & \y & \y & \cellcolor{DarkRowRiver}\y \\
    
    \rowcolor{DarkRow}
    IPDA & \n & \n & \n & \n & \cellcolor{LightRowRiver}\hphantom{*}\y* \\
    \end{tabular}%
    }
    \caption{Overview of the properties for Split Cycle and refinements.
    ``\y'' indicates the property is satisfied for general profiles.
    ``\y*'' indicates the property is satisfied for general profiles with a suitable tiebreaker.
    ``$\y^\circ$'' indicates the property is satisfied for uniquely weighted profiles.
    ``\n'' indicates the property is violated even for uniquely weighted profiles.
    For all results not concerning River or IPDA, we refer to the paper by \citet{holliday2020split}.}
    \label{tab:BasicAxioms}
    \end{table}

In this section, we analyze the axiomatic properties of River and compare them to the other introduced social choice functions. An overview of our results is provided in \Cref{tab:BasicAxioms}.

% \subsection{Basic Axioms} 
    We start by observing that River satisfies all basic axioms required of a reliable social choice function. 
    For formal definitions and proofs, we refer to the \Cref{subsec:app_axioms} in the appendix.

    %%%%%%%%%%%%%%%%%%%%%%% Anonymity and neutrality %%%%%%%%%%%%%%%%%%%%%%%
    \defstyle{Anonymity} and \defstyle{neutrality} require that all  voters, respectively alternatives, are treated equally.
    Both properties are naturally satisfied by River for uniquely weighted preference profiles.
    In profiles that are not uniquely weighted, however, the tiebreaker might invalidate anonymity or neutrality.\footnote{Ranked Pairs and Stable Voting encounter the same issue. Generally, anonymity and neutrality conflict with resoluteness; e.\,g., consider a profile with $A_\Pref=\{x,y\}$ and $m_\Pref(x,y)=0$.}
    
    % %%%%%%%%%%%%%%%%%%%%%%% Reversal Symmetry %%%%%%%%%%%%%%%%%%%%%%%
    % A natural extension of neutrality is \defstyle{reversal symmetry}.
    % %While neutrality is concerned with pairwise switching alternatives' places in every ballot, reversal symmetry extends this to switching the whole ballot of each voter.
    % It states that if there is exactly one winner for a profile, this alternative should not be a winner for fully reversed preferences.
    % This is satisfied by Split Cycle and all its refinements.
    % % Note that this is the case for any kind of preferences. The assumption of reversal symmetry, implying a unique winning alternative, ensures that no tiebreaker for non-uniquely weighted preferences has negative effects.
    
    %%%%%%%%%%%%%%%%%%%%%%% Monotonicity %%%%%%%%%%%%%%%%%%%%%%%
    \defstyle{Monotonicity} demands that if support for a winning alternative increases (\ie some voters rank it higher without changing the relative order of the other alternatives), this alternative must remain a winner.  
    River satisfies monotonicity, as do Split Cycle, Ranked Pairs, and Beat Path. Maybe surprisingly, Stable Voting violates monotonicity \cite{holliday2023stable}.

    %%%%%%%%%%%%%%%%%%%%%%%% condorcet %%%%%%%%%%%%%%%%%%%%%%%
    River, like the other social choice functions, always selects the Condorcet winner if one exists, and never selects a Condorcet loser.  
    %%%%%%%%%%%%%%%%%%%%%%% Smith %%%%%%%%%%%%%%%%%%%%%%%
    Recall the \defstyle{Smith set} (\Cref{sec:immu}) as a generalization of a Condorcet winner. 
    The \defstyle{Smith criterion} states that the set of winners has to be chosen from the Smith set. 
    This property is implied by selecting only immune alternatives and thus satisfied by all Split Cycle refinements.
    
    %%%%%%%%%%%%%%%%%%%%%%% Pareto %%%%%%%%%%%%%%%%%%%%%%%
    Given a preference profile $\Pref$ and two alternatives $x,y \in A$, $y$ \textit{Pareto-dominates} $x$ if every voter ranks $y$ over $x$, i.e., $y \succ_i x$ for all $i \in N$.
    \textit{Pareto efficiency} requires that Pareto-dominated alternatives are never chosen. Since Split Cycle satisfies Pareto efficiency, so do all its refinements. 
    
\subsection{Independence of Smith-Dominated Alternatives} \label{subsec:ISDA}

Next, we consider independence criteria. These properties prescribe that the winning set of a social choice function should not be affected by the presence (or absense) of alternatives that are in some sense ``inferior''.
Independence properties can also be interpreted as  safeguards against agenda manipulation: introducing an inferior alternative into an election should not disturb the set of chosen candidates. 

    % An independence property consists of two parts:  
    % First, it defines a classification for a type of alternative. Second, it is satisfied only if the winning set remains unchanged when an alternative of that type is removed from the election.

    %%%%%%%%%%%%%%%%%%%%%%% ISDA %%%%%%%%%%%%%%%%%%%%%%%
    We begin with independence of Smith-dominated alternatives.
    Recall that the Smith set is the smallest dominating subset of alternatives.
    Every alternative \textit{not} in the Smith set is called \defstyle{Smith-dominated}.
   {Independence of Smith-dominated alternatives} requires that the winners do not change when a Smith-dominated alternative is removed.
    \begin{definition}%[ISDA]
    \label{def:ISDA}
        A social choice function \SCF is \defstyle{independent of Smith-dominated alternatives} (ISDA) if for any preference profile \Pref and $x\in A\setminus \SM(\Pref)$, we have $\SCF(\Pref)=\SCF(\Pref_{-x})$.
    \end{definition}\noindent
     Split Cycle, Ranked Pairs, and Beat Path satisfy ISDA \cite{holliday2020split}. 
     Stable Voting satisfied ISDA for uniquely weighted profiles, but violates it in general \cite{holliday2023stable}.
    We show, River satisfies ISDA. % as long as the employed tiebreaker is consistent.

    \begin{theorem}
         River satisfies ISDA for general preference profiles. % when equipped with a consistent tiebreaker.
    \end{theorem}
    \begin{proof}
    Let $x\in A\setminus \SM(\Pref)$.
    By definition of the Smith set, there can be no majority cycle in \margingraph\ involving an element of the Smith set and $x$ at the same time.
    Towards $\RV(\Pref)\subseteq\RV(\Pref_{-x})$, let $y\in\RV(\Pref)$ and assume towards contradiction $y\not\in\RV(\Pref_{-x})$. Then $y\in \SM(\Pref)$ and there is an edge $(z,y)\in E(\margingraphRVExt{\Pref_{-x}})$ for some $z\in A\setminus\{x,y\}$, but $(z,y)\notin E(\margingraphRVExt{\Pref})$.
    Since the tiebreaking order for $\Pref_{-x}$ is the order of \Pref without all edges containing $x$, $(z,y)\notin E(\margingraphRVExt{\Pref})$ can only be because there is a majority cycle in $\margingraph(\Pref)$ with $(z,y)$ as its splitting edge that is not in $\margingraph(\Pref_{-x})$. But $x$ cannot be part of a majority cycle containing also $y$.
    The other direction is analogous.
    %\markus{Der konsitente tiebreaker sollte im Beweis vorkommen.}
    \end{proof}

    %%%%%%%%%%%%%%%%%%%%%%% IPDA %%%%%%%%%%%%%%%%%%%%%%%

    \subsection{Independence of Pareto-Dominated Alternatives} \label{subsec:IPDA}

    Next, we consider alternatives that are ``inferior'' because they are \textit{Pareto}-dominated.
    The resulting property can be defined analogously to ISDA.\footnote{IPDA is logically independent from ISDA, because neither does Pareto-dominance imply Smith-dominance,  nor vice versa \citep{fishburn_theory_1973}.}
    Independence of Pareto-dominated alternatives was already studied by Fishburn \cite{fishburn_theory_1973} under the name ``reduction condition.''
    \begin{definition}%[IPDA]
    \label{def:IPDA}
        A social choice function \SCF is \defstyle{independent of Pareto-dominated alternatives} (IPDA) if for any preference profile \Pref and $x,y\in A_\Pref$ with $\marginExt{y,x}{\Pref} = |N_\Pref|$, we have $\SCF(\Pref)=\SCF(\Pref_{-x})$.
    \end{definition}\noindent

    Observe that constructing a Pareto-dominated alternative is not a complex endeavour. One can choose any of the current alternatives as a blueprint and construct a copy alternative which is worse in every aspect.
    
    Surprisingly, Split Cycle, Ranked Pairs, Beat Path and Stable Voting all violate IPDA, even for uniquely weighted preferences with at most $5$ alternatives.\footnote{The example for Stable Voting is due to Wesley Holliday (personal communication, 2024).}

    \begin{theorem} \label{thm:IPDA_otherSC}
        Split Cycle, Ranked Pairs, Beat Path and Stable Voting do not satisfy IPDA, not even for uniquely weighted profiles.
    \end{theorem}
    \begin{proof} 
    We present the counterexamples as margin graphs in  \Cref{fig:IPDA_otherSC}. The corresponding preference profiles $\Pref_1$, $\Pref_2$, $\Pref_3$, and $\Pref_4$ can be found in \Cref{subsec:app_profiles}. Note that in all profiles, $b$ Pareto-dominates $a$.
%Note that all counterexamples are uniquely weighted.
    
        \textit{Split Cycle:} Observe that $c\in\SC(\Pref_1)$, because its only incoming edge $(d,c)$ is the smallest-margin edge of cycle $(c, a, d, c)$. However, $c\notin\SC((\Pref_1)_{-a})$, because $(c,a,d,c)$ was the only cycle that $(d,c)$ was contained in.

        \textit{Ranked Pairs:} We observe, $\RP(\Pref_2)=\{d\}$. Ranked Pairs adds $(b,a),(b,c),(a,c),(d,a)$, skips $(c,d)$ because of the cycle $(c,d,a,c)$, adds $(d,a)$, and skips $(d,b)$. However, $\RP((\Pref_2)_{-a})=\{b\}$, because without $a$, Ranked Pairs adds $(b,c)$ and $(c,d)$, skipping $(d,b)$.
    
        \textit{Beat Path:} It can be verified by computing the beatpath strengths that $\BP(\Pref_3)=\{d\}$, but $\BP((\Pref_3)_{-a})=\{c\}$.
    
        \textit{Stable Voting:} Again, it can be verified that $\SV(\Pref_4)=\{d\}$, but $\SV((\Pref_4)_{-a})=\{c\}$. 
    \end{proof}   
    
    \begin{figure}[t]
        \centering
        \includegraphics[width=0.9\textwidth]{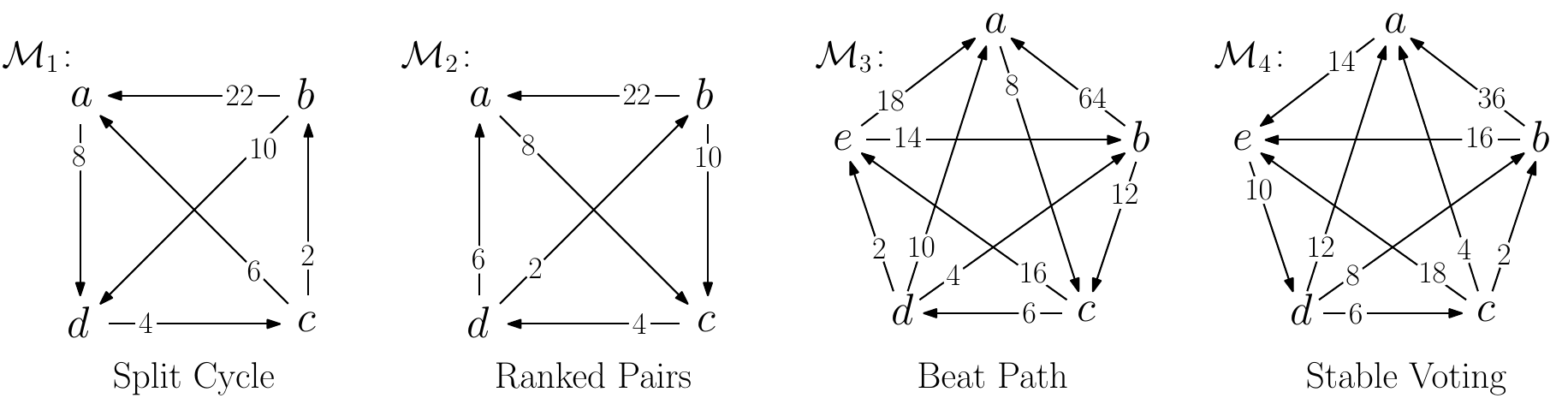}
        \caption{Margin graphs for the preference profiles used in the proof of \Cref{thm:IPDA_otherSC}. 
        % (first three graphs from the left) and \Cref{thm:IPDA_SV} (rightmost graph). 
        % Corresponding preference profiles can be found in the appendix. 
        }
        \label{fig:IPDA_otherSC}
    \end{figure}

    River satisfies IPDA due to its unique way of resolving cycles in the margin graph. We prove this claim for uniquely weighted preference profiles in \Cref{thm:IPDA_River_uw}, and extend it in \Cref{thm:IPDA_River_TB} to general profiles. % as long as River is equipped with a tiebreaker satisfying a simple consistency property. 
    
    To build towards the first result, we start with the following observation related to the concept of covering:
    % First, we make the following observation which will be used in \Cref{thm:IPDA_River_uw}.
    % Given a preference profile \Pref and two alternatives $x,y \in A_\Pref$,  $y$ is said to \defstyle{cover}  $x$ if $\margin{y,x}>0$ and $\margin{y,z}\geq\margin{x,z}$ for all $z\in A\setminus\{x,y\}$. 
    An alternative $y$ is said to \defstyle{cover} another alternative $x$ if $\margin{y,x}>0$ and $\margin{y,z}\geq\margin{x,z}$ for all $z\in A\setminus\{x,y\}$. We observe that Pareto domination implies covering.
    \begin{observation}%[\cite{Bord83a}]
        Given a preference profile \Pref, if $y$ Pareto-dominates $x$, then $y$ covers $x$.
    \end{observation}
    Intuitively, if $y$ Pareto-dominates $x$, then $\margin{y,x}=n$ and $y$ is ranked above $x$ by every voter. Thus, whenever $x\pref_i z$, we also have $y\pref_i x\pref_i z$. This observation will be instrumental in the proof of the following claim.
    \begin{theorem} \label{thm:IPDA_River_uw}
        River satisfies IPDA for uniquely weighted preference profiles.
    \end{theorem}
    \newcommand{\EnoX}{E_{-x}}
    \newcommand{\uw}{uniquely weighted\xspace}
    \begin{proof}
        Let \Pref be a uniquely weighted preference profile and let $x,y \in A_\Pref$ such that $y$ Pareto-dominates $x$. Let $E=E(\margingraphRVExt{\Pref})$ and $\EnoX=E(\margingraphRVExt{\Pref_{-x}})$ denote the edge sets of the River diagrams with and without $x$, respectively. %Then $y$ covers $x$. %By \Cref{lem:Pareto_cover}, $y$ covers $x$. 
        
        First, we show that $(y,x)\in E$. Since $\margin{y,x}=n$ and the preference profile is \uw, the edge $(y,x)$ is processed as the first edge. 
        Hence, adding $(y,x)$ can neither form a cycle nor branching, and $(y,x)$ is added to $E$.
        
        Next, we show that $(x,z)\notin E$ for all $z\in A\setminus\{x\}$.
        Assume towards contradiction that $(x,z)$ is added to $E$. Then $(Br)$ was not fulfilled, and thus $(y,z)\notin E$.
        Since $y$ covers $x$ and \Pref is \uw, ${\margin{y,z}>\margin{x,z}}$ and $(y,z)$ is processed before $(x,z)$. 
        If $(y,z)$ was not added because of $(Br)$, \ie there is some $(z',z)\in E$ with higher margin, then $(Br)$ would also reject $(x,z)$. This would be a contradiction.
        Therefore, $(y,z)$ must be rejected by $(Cy)$. This means, there must be a path $(z,p_1,\dots,p_k,y)\in E$ with strength larger than $\margin{y,z}$.
        But this path would form a cycle with $(y,x)$ and $(x,z)$: $(z,p_1,\dots,p_k,y,x,z)$. Thus, $(c,y)$ would reject $(x,z)$, a contradiction. We conclude that $(x,z)\notin E$.

        So in $\Pref$, no edge adjacent to $x$ is considered in $E$ apart from $(y,x)$, which cannot be part of any cycle, since $x$ has no outgoing edges in $E$. Removing $x$ from the election does not influence the margin of any edge not containing $x$, and all edges incident to $x$ are simply removed. Therefore, in $\EnoX$, we have an edge $(z,z')\in\EnoX$ if and only if $(z,z')\in E$.
        % Finally, we show that all edges not containing $x$ are in $E$ if and only if they are in $\EnoX$. That is, for all $x \notin \{z,z'\}$, we have
        % (*) $(z,z')\in E$ if and only if $(z,z')\in\EnoX$.
        % % \begin{align} \tag{$*$}
        % %     % (*) \quad \text{for all }z\neq x\neq z'\text{, we have }
        % %     (z,z')\in E \quad \text{if and only if} \quad (z,z')\in\EnoX. \label{star}
        % % \end{align}
        % We prove this by induction over $(z,z')$ from largest to smallest margin. Assume we are about to process $(z,z')$, and (*) is true for all edges processed so far. 
        % First, assume $(z,z')\in E\setminus \EnoX$. Then we must have added some $(z'',z')$ or some path $z'\to w_1\cdots\to w_k\to z$ with higher strength to $\EnoX$. Since then $z''\neq x$ and all $w_i\neq x$, it follows from (\ref{star}) that the same edges $(z'',z')$, respectively $z'\to w_1\cdots\to w_k\to z$, have also been added to $E$, hence $(z,z')\notin E$ after all. 
        % Assume, on the other hand, $(z,z')\in \EnoX\setminus E$. Then the same argument holds with $E$ and $\EnoX$ exchanged, so $(z,z')\notin \EnoX$ after all. 
        % Hence $(z,z')$ is in either both $E$ and $\EnoX$ or neither, proving the induction step.
        
        Therefore, all edges in $E$ apart from $(y,x)$ are in $\EnoX$, and thus $\RV(\Pref)=\RV(\Pref_{-x})$.
    \end{proof}

    \newcommand{\parecon}{Pareto-consistent\xspace}
    
    For non-uniquely weighted preference profiles, River satisfies IPDA when equipped with any tiebreaker that ``respects'' Pareto dominance, which we call \textit{\parecon}. 
    % considers Pareto-dominating alternatives' defeats first. 
    % We call such tiebreakers \defstyle{\parecon}.
    \begin{definition}
        A consistent tiebreaker is called \defstyle{\parecon} if, whenever $y$ Pareto-dominates $x$, for all $z\notin\{x,y\}$, the edge $(y,z)$ is ranked higher than the edge $(x,z)$. 
    \end{definition}
    Observe that such a \parecon tiebreaker always exists: lexicographical tiebreakers are \parecon since they are consistent and the tiebreaking voter ranks $y$ above $x$.
    %, even when there a multiple Pareto dominations: If an alternative $y'$ Pareto-dominates alternative $y$, which itself Pareto-dominates alternative $x$, then $y'$ also Pareto-dominates $x$. Therefore, for any alternative $z\in A_\Pref\setminus\{y',y,x\}$, a Pareto-consistent tiebreaker has to rank $(y',z)$ higher than $(y,z)$, which has to be ranked higher than $(x,z)$.
    % \setcounter{theorem}{6}
    % \setcounter{theorem}{7}
    \begin{restatable}{theorem}{thmIPDARiverTB}
    {\restatemarker}
    \label{thm:IPDA_River_TB}
        River satisfies IPDA when equipped with a \parecon tiebreaker.
    % \end{theorem}
    \end{restatable}

    The proof can be found in \Cref{subsec:app_IPDA}, where we also show that we cannot drop the Pareto-consistency requirement by providing an example where River fails to satisfy IPDA for some tiebreaker that is not \parecon. 
    
   We can strengthen \Cref{thm:IPDA_River_TB} by considering a more permissive notion of dominance that we call \textit{quasi-Pareto dominance}. Instead of requiring that $y$ is preferred over $x$ by every voter (i.e., $\margin{y,x}=n$), we merely require that $y$ covers $x$ and that the margin of $y$ over $x$ is stronger than any other margin involving $x$. %in some sense extension of cover consistency
    %    
    % Arguably, in large elections Pareto-dominated alternatives are very rare and hard to manufacture. We therefore suggest the following weaker notion of quasi-Pareto domination, and the corresponding axiom of independence of quasi-Pareto-dominated alternatives:
        \begin{definition}%[Quasi-Pareto domination]
            Let $\Pref$ be a  profile and $x,y\in A_\Pref$.
            We say \defstyle{$y$ quasi-Pareto-dominates $x$} if (1) $y$ covers $x$ and (2)
            for all $z\in A_\Pref\setminus\{x,y\}$, %we have 
            % covering:
            % $\margin{y,x}>0$, 
            % $\margin{y,z}\ge\margin{x,z}$, 
            $\margin{y,x}\ge\margin{z,x}$ and $\margin{y,x}\ge\margin{x,z}$.
            % $\marginExt{y,x}{\Pref}\ge\marginExt{z,x}{\Pref}$ and $\marginExt{y,x}{\Pref}\ge\marginExt{x,z}{\Pref}$.
            % and for all $y'\in A_\Pref$ with $\margin{y',x}>\margin{y,x}$, we also have $\margin{y',z}\ge\margin{x,z}$.
        \end{definition} 
    It is easy to check that Pareto domination implies quasi-Pareto domination and that the latter is an acyclic relation. 
    This leads to a strengthened version of IPDA, namely \textit{inpedendence of quasi-Pareto-dominated alternatives} (IQDA).
%Moreover, if $y$ quasi-Pareto-dominates $x$ and $\margin{y',x}>\margin{y,x}$, then also $y'$ quasi-Pareto-dominates $x$.
        % \begin{definition}%[IQDA] 
        % \label{def:IQDA}
        %     A social choice function \SCF is \defstyle{independent of quasi-Pareto-dominated alternatives} (IQDA) if for any preference profile \Pref and $x,y\in A_\Pref$ where $y$ quasi-Pareto-dominates $x$, we have $\SCF(\Pref)=\SCF(\Pref_{-x})$.
        % \end{definition}
    We show that River satisfies IQDA when it is equipped with quasi-Pareto-consistent tiebreakers, which are defined analogously to Pareto-consistent tiebreakers. 
    %The proof follows the same technique as the proof for IPDA. %by considering quasi-Pareto-consistent tiebreakers that .... 
    The formal definition of such tiebreakers and the proof of the following theorem can be found in \Cref{subsec:app_IPDA}.
    % \textcolor{blue}{``We can show that River statisfies \textit{inpedendence of quasi-Pareto-dominated alternatives} (IQDA) in a similar fashion as for IPDA by considering quasi-Pareto-consistent tiebreakers that .... The definition of such tiebreakers and the proof of the following theorem can be found in the appendix.''}
    %     \begin{definition}
    %         A \parecon tiebreaker is called \defstyle{quasi-Pareto-consistent} if for two edges $(y,x)$, $(y',x)$ with equal margins, of which $y$ but not $y'$ quasi-Pareto-dominate $x$, the tiebreaker ranks $(y,x)$ before $(y',x)$, and whenever $y$ quasi-Pareto-dominates $x$, any edge $(x,z)$ is always ranked after both the edges $(y,z)$ and $(y,x)$.
    %     \end{definition}
    % Note that quasi-Pareto-consistent tiebreakers exist since quasi-Pareto domination is acyclic.

    \begin{restatable}{theorem}{thmIQDARiver}
    {\restatemarker}
    % \begin{theorem}[$\star$]
    \label{thm:IQDA_River}
        River satisfies IQDA for \uw preference profiles, and for general profiles when equipped with a quasi-Pareto-consistent tiebreaker. %, also for non-\uw preference profiles.
   \end{restatable}

\section{Discussion}
We introduced River, a novel single-winner voting method based on majority margins, and compared it to established methods including Ranked Pairs, Beat Path, Stable Voting, and Split Cycle. Like these methods, River is a refinement of Split Cycle—but with a greatly simplified decision process that resembles Ranked Pairs while avoiding many of its complications.

Our axiomatic analysis showed that River satisfies a wide range of desirable properties, including key independence criteria that guard against agenda manipulation. In particular, River satisfies independence of Smith-dominated alternatives and, crucially, independence of Pareto-dominated alternatives (IPDA)\,--\,a property not satisfied by any of the other methods we considered.

% Beyond the simplicity of Rivers procedure, the tree structure of its rebutting diagram enables powerful computational advantages. Most notably, it was recently shown in \cite{riverPUT2024melanowski} that the parallel-universe variant of River (\mbox{\RV-$PUT$})-which returns all winner that could arise under any valid tiebreaking order—can be computed in polynomial time. This is in stark contrast to the equivalent problem for Ranked Pairs, which is known to be \NP-complete. Crucially, the tractability of \RV-$PUT$ highly depends on the tree structure of River.

Overall, River appears to be a robust and transparent voting method that offers some axiomatic advantages over established voting methods, is straightforward to compute (in particular ``by hand''), and generates simple, easy-to-interpret diagrams that justify the winner's selection.

% \medskip\noindent \textbf{Open Questions.}\quad 
% Some of River’s properties—such as its resistance to strategic voting—depend on the choice of the tiebreaker. While our results show that River satisfies different desirable axioms under different tiebreaking assumptions, it remains open whether a single tiebreaker exists that simultaneously satisfies all of them. Future work should investigate this question, along with a broader axiomatic and empirical comparison to other methods.

% Additionally, the practical computational efficiency of River in comparison to the other voting methods is still open. Not only in regards to computational running time, but also efficiency of manipulation.

\backmatter
% \bmhead{Acknowledgements}

% Acknowledgements are not compulsory. Where included they should be brief. Grant or contribution numbers may be acknowledged.

% Please refer to Journal-level guidance for any specific requirements.
\bibliography{algo,aaai25, abb} %, zotero}% common bib file


\begin{thebibliography}{32}
\ifx \bisbn   \undefined \def \bisbn  #1{ISBN #1}\fi
\ifx \binits  \undefined \def \binits#1{#1}\fi
\ifx \bauthor  \undefined \def \bauthor#1{#1}\fi
\ifx \batitle  \undefined \def \batitle#1{#1}\fi
\ifx \bjtitle  \undefined \def \bjtitle#1{#1}\fi
\ifx \bvolume  \undefined \def \bvolume#1{\textbf{#1}}\fi
\ifx \byear  \undefined \def \byear#1{#1}\fi
\ifx \bissue  \undefined \def \bissue#1{#1}\fi
\ifx \bfpage  \undefined \def \bfpage#1{#1}\fi
\ifx \blpage  \undefined \def \blpage #1{#1}\fi
\ifx \burl  \undefined \def \burl#1{\textsf{#1}}\fi
\ifx \doiurl  \undefined \def \doiurl#1{\url{https://doi.org/#1}}\fi
\ifx \betal  \undefined \def \betal{\textit{et al.}}\fi
\ifx \binstitute  \undefined \def \binstitute#1{#1}\fi
\ifx \binstitutionaled  \undefined \def \binstitutionaled#1{#1}\fi
\ifx \bctitle  \undefined \def \bctitle#1{#1}\fi
\ifx \beditor  \undefined \def \beditor#1{#1}\fi
\ifx \bpublisher  \undefined \def \bpublisher#1{#1}\fi
\ifx \bbtitle  \undefined \def \bbtitle#1{#1}\fi
\ifx \bedition  \undefined \def \bedition#1{#1}\fi
\ifx \bseriesno  \undefined \def \bseriesno#1{#1}\fi
\ifx \blocation  \undefined \def \blocation#1{#1}\fi
\ifx \bsertitle  \undefined \def \bsertitle#1{#1}\fi
\ifx \bsnm \undefined \def \bsnm#1{#1}\fi
\ifx \bsuffix \undefined \def \bsuffix#1{#1}\fi
\ifx \bparticle \undefined \def \bparticle#1{#1}\fi
\ifx \barticle \undefined \def \barticle#1{#1}\fi
\bibcommenthead
\ifx \bconfdate \undefined \def \bconfdate #1{#1}\fi
\ifx \botherref \undefined \def \botherref #1{#1}\fi
\ifx \url \undefined \def \url#1{\textsf{#1}}\fi
\ifx \bchapter \undefined \def \bchapter#1{#1}\fi
\ifx \bbook \undefined \def \bbook#1{#1}\fi
\ifx \bcomment \undefined \def \bcomment#1{#1}\fi
\ifx \oauthor \undefined \def \oauthor#1{#1}\fi
\ifx \citeauthoryear \undefined \def \citeauthoryear#1{#1}\fi
\ifx \endbibitem  \undefined \def \endbibitem {}\fi
\ifx \bconflocation  \undefined \def \bconflocation#1{#1}\fi
\ifx \arxivurl  \undefined \def \arxivurl#1{\textsf{#1}}\fi
\csname PreBibitemsHook\endcsname

%%% 1
\bibitem[\protect\citeauthoryear{Arrow et~al.}{2010}]{arrow_handbook_2010}
\begin{bbook}
\bauthor{\bsnm{Arrow}, \binits{K.J.}},
\bauthor{\bsnm{Sen}, \binits{A.}},
\bauthor{\bsnm{Suzumura}, \binits{K.}}:
\bbtitle{Handbook of Social Choice and Welfare}
vol. \bseriesno{2}.
\bpublisher{Elsevier},
\blocation{Amsterdam}
(\byear{2010})
\end{bbook}
\endbibitem

%%% 2
\bibitem[\protect\citeauthoryear{Brandt et~al.}{2016}]{brandt_handbook_2016}
\begin{bbook}
\beditor{\bsnm{Brandt}, \binits{F.}},
\beditor{\bsnm{Conitzer}, \binits{V.}},
\beditor{\bsnm{Endriss}, \binits{U.}},
\beditor{\bsnm{Lang}, \binits{J.}},
\beditor{\bsnm{Procaccia}, \binits{A.D.}} (eds.):
\bbtitle{Handbook of Computational Social Choice}.
\bpublisher{Cambridge University Press},
\blocation{Cambridge}
(\byear{2016})
\end{bbook}
\endbibitem

%%% 3
\bibitem[\protect\citeauthoryear{F{\"u}rnkranz and H{\"u}llermeier}{2003}]{furnkranz2003pairwise}
\begin{bchapter}
\bauthor{\bsnm{F{\"u}rnkranz}, \binits{J.}},
\bauthor{\bsnm{H{\"u}llermeier}, \binits{E.}}:
\bctitle{Pairwise preference learning and ranking}.
In: \bbtitle{Machine Learning: ECML 2003},
pp. \bfpage{145}--\blpage{156}
(\byear{2003}).
\bcomment{Springer}
\end{bchapter}
\endbibitem

%%% 4
\bibitem[\protect\citeauthoryear{Askell et~al.}{2021}]{askell2021general}
\begin{botherref}
\oauthor{\bsnm{Askell}, \binits{A.}},
\oauthor{\bsnm{Bai}, \binits{Y.}},
\oauthor{\bsnm{Chen}, \binits{A.}},
\oauthor{\bsnm{Drain}, \binits{D.}},
\oauthor{\bsnm{Ganguli}, \binits{D.}},
\oauthor{\bsnm{Henighan}, \binits{T.}},
\oauthor{\bsnm{Jones}, \binits{A.}},
\oauthor{\bsnm{Joseph}, \binits{N.}},
\oauthor{\bsnm{Mann}, \binits{B.}},
\oauthor{\bsnm{DasSarma}, \binits{N.}}, et al.:
A general language assistant as a laboratory for alignment.
arXiv preprint arXiv:2112.00861
(2021)
\end{botherref}
\endbibitem

%%% 5
\bibitem[\protect\citeauthoryear{K{\"o}pf et~al.}{2024}]{kopf2024openassistant}
\begin{botherref}
\oauthor{\bsnm{K{\"o}pf}, \binits{A.}},
\oauthor{\bsnm{Kilcher}, \binits{Y.}},
\oauthor{\bsnm{R{\"u}tte}, \binits{D.}},
\oauthor{\bsnm{Anagnostidis}, \binits{S.}},
\oauthor{\bsnm{Tam}, \binits{Z.R.}},
\oauthor{\bsnm{Stevens}, \binits{K.}},
\oauthor{\bsnm{Barhoum}, \binits{A.}},
\oauthor{\bsnm{Nguyen}, \binits{D.}},
\oauthor{\bsnm{Stanley}, \binits{O.}},
\oauthor{\bsnm{Nagyfi}, \binits{R.}}, et al.:
Openassistant conversations-democratizing large language model alignment.
Advances in Neural Information Processing Systems
\textbf{36}
(2024)
\end{botherref}
\endbibitem

%%% 6
\bibitem[\protect\citeauthoryear{Mishra}{2023}]{mishra2023ai}
\begin{botherref}
\oauthor{\bsnm{Mishra}, \binits{A.}}:
{AI} alignment and social choice: Fundamental limitations and policy implications.
arXiv preprint arXiv:2310.16048
(2023)
\end{botherref}
\endbibitem

%%% 7
\bibitem[\protect\citeauthoryear{Zwicker}{2016}]{zwicker2016introduction}
\begin{botherref}
\oauthor{\bsnm{Zwicker}, \binits{W.S.}}:
Introduction to the theory of voting.
Handbook of computational social choice
\textbf{2}
(2016)
\end{botherref}
\endbibitem

%%% 8
\bibitem[\protect\citeauthoryear{May}{1952}]{May52a}
\begin{barticle}
\bauthor{\bsnm{May}, \binits{K.}}:
\batitle{A set of independent, necessary and sufficient conditions for simple majority decisions}.
\bjtitle{Econometrica}
\bvolume{20}(\bissue{4}),
\bfpage{680}--\blpage{684}
(\byear{1952})
\end{barticle}
\endbibitem

%%% 9
\bibitem[\protect\citeauthoryear{Brandt et~al.}{2016}]{brandt_tournament_2016}
\begin{bchapter}
\bauthor{\bsnm{Brandt}, \binits{F.}},
\bauthor{\bsnm{Brill}, \binits{M.}},
\bauthor{\bsnm{Harrenstein}, \binits{P.}}:
\bctitle{Tournament solutions}.
In: \beditor{\bsnm{Brandt}, \binits{F.}},
\beditor{\bsnm{Conitzer}, \binits{V.}},
\beditor{\bsnm{Endriss}, \binits{U.}},
\beditor{\bsnm{Lang}, \binits{J.}},
\beditor{\bsnm{Procaccia}, \binits{A.D.}} (eds.)
\bbtitle{Handbook of Computational Social Choice}.
\bpublisher{Cambridge University Press},
\blocation{Cambridge}
(\byear{2016}).
\bcomment{Chap. 3}
\end{bchapter}
\endbibitem

%%% 10
\bibitem[\protect\citeauthoryear{Laslier}{1997}]{Lasl97a}
\begin{bbook}
\bauthor{\bsnm{Laslier}, \binits{J.-F.}}:
\bbtitle{Tournament Solutions and Majority Voting},
(\byear{1997})
\end{bbook}
\endbibitem

%%% 11
\bibitem[\protect\citeauthoryear{Schwartz}{1986}]{schwartz_logic_1986}
\begin{bbook}
\bauthor{\bsnm{Schwartz}, \binits{T.}}:
\bbtitle{The Logic of Collective Choice}.
\bpublisher{Columbia University Press},
\blocation{New York}
(\byear{1986})
\end{bbook}
\endbibitem

%%% 12
\bibitem[\protect\citeauthoryear{Good}{1971}]{Good71a}
\begin{barticle}
\bauthor{\bsnm{Good}, \binits{I.J.}}:
\batitle{A note on {C}ondorcet sets}.
\bjtitle{Public Choice}
\bvolume{10}(\bissue{1}),
\bfpage{97}--\blpage{101}
(\byear{1971})
\end{barticle}
\endbibitem

%%% 13
\bibitem[\protect\citeauthoryear{Smith}{1973}]{Smit73a}
\begin{barticle}
\bauthor{\bsnm{Smith}, \binits{J.H.}}:
\batitle{Aggregation of preferences with variable electorate}.
\bjtitle{Econometrica}
\bvolume{41}(\bissue{6}),
\bfpage{1027}--\blpage{1041}
(\byear{1973})
\end{barticle}
\endbibitem

%%% 14
\bibitem[\protect\citeauthoryear{Fischer et~al.}{2016}]{fischer_weighted_2016}
\begin{bchapter}
\bauthor{\bsnm{Fischer}, \binits{F.}},
\bauthor{\bsnm{Hudry}, \binits{O.}},
\bauthor{\bsnm{Niedermeier}, \binits{R.}}:
\bctitle{Weighted tournament solutions}.
In: \beditor{\bsnm{Brandt}, \binits{F.}},
\beditor{\bsnm{Conitzer}, \binits{V.}},
\beditor{\bsnm{Endriss}, \binits{U.}},
\beditor{\bsnm{Lang}, \binits{J.}},
\beditor{\bsnm{Procaccia}, \binits{A.D.}} (eds.)
\bbtitle{Handbook of Computational Social Choice}.
\bpublisher{Cambridge University Press},
\blocation{Cambridge}
(\byear{2016}).
\bcomment{Chap. 4}
\end{bchapter}
\endbibitem

%%% 15
\bibitem[\protect\citeauthoryear{Heitzig}{2002}]{heitzig2002social}
\begin{botherref}
\oauthor{\bsnm{Heitzig}, \binits{J.}}:
Social choice under incomplete, cyclic preferences.
arXiv preprint math/0201285
(2002)
\end{botherref}
\endbibitem

%%% 16
\bibitem[\protect\citeauthoryear{Dung}{1995}]{dung1995acceptability}
\begin{barticle}
\bauthor{\bsnm{Dung}, \binits{P.M.}}:
\batitle{On the acceptability of arguments and its fundamental role in nonmonotonic reasoning, logic programming and n-person games}.
\bjtitle{Artificial intelligence}
\bvolume{77}(\bissue{2}),
\bfpage{321}--\blpage{357}
(\byear{1995})
\end{barticle}
\endbibitem

%%% 17
\bibitem[\protect\citeauthoryear{Holliday and Pacuit}{2023}]{holliday2020split}
\begin{barticle}
\bauthor{\bsnm{Holliday}, \binits{W.H.}},
\bauthor{\bsnm{Pacuit}, \binits{E.}}:
\batitle{Split cycle: a new {C}ondorcet-consistent voting method independent of clones and immune to spoilers}.
\bjtitle{Public Choice}
\bvolume{197},
\bfpage{1}--\blpage{62}
(\byear{2023})
\end{barticle}
\endbibitem

%%% 18
\bibitem[\protect\citeauthoryear{Tideman}{1987}]{tideman1987independence}
\begin{barticle}
\bauthor{\bsnm{Tideman}, \binits{T.N.}}:
\batitle{Independence of clones as a criterion for voting rules}.
\bjtitle{Social Choice and Welfare}
\bvolume{4},
\bfpage{185}--\blpage{206}
(\byear{1987})
\end{barticle}
\endbibitem

%%% 19
\bibitem[\protect\citeauthoryear{Schulze}{2011}]{schulze2011new}
\begin{barticle}
\bauthor{\bsnm{Schulze}, \binits{M.}}:
\batitle{A new monotonic, clone-independent, reversal symmetric, and {C}ondorcet-consistent single-winner election method}.
\bjtitle{Social choice and Welfare}
\bvolume{36},
\bfpage{267}--\blpage{303}
(\byear{2011})
\end{barticle}
\endbibitem

%%% 20
\bibitem[\protect\citeauthoryear{Holliday and Pacuit}{2023}]{holliday2023stable}
\begin{botherref}
\oauthor{\bsnm{Holliday}, \binits{W.H.}},
\oauthor{\bsnm{Pacuit}, \binits{E.}}:
Stable voting.
Constitutional Political Economy,
421--433
(2023)
\end{botherref}
\endbibitem

%%% 21
\bibitem[\protect\citeauthoryear{Fishburn}{1973}]{fishburn_theory_1973}
\begin{bbook}
\bauthor{\bsnm{Fishburn}, \binits{P.C.}}:
\bbtitle{The Theory of Social Choice}.
\bpublisher{Princeton University Press},
\blocation{Princeton}
(\byear{1973})
\end{bbook}
\endbibitem

%%% 22
\bibitem[\protect\citeauthoryear{Richelson}{1978}]{Rich78c}
\begin{barticle}
\bauthor{\bsnm{Richelson}, \binits{J.}}:
\batitle{A characterization result for the plurality rule}.
\bjtitle{Journal of Economic Theory}
\bvolume{19}(\bissue{2}),
\bfpage{548}--\blpage{550}
(\byear{1978})
\end{barticle}
\endbibitem

%%% 23
\bibitem[\protect\citeauthoryear{Ching}{1996}]{Chin96a}
\begin{barticle}
\bauthor{\bsnm{Ching}, \binits{S.}}:
\batitle{A simple characterization of plurality rule}.
\bjtitle{Journal of Economic Theory}
\bvolume{71}(\bissue{1}),
\bfpage{298}--\blpage{302}
(\byear{1996})
\end{barticle}
\endbibitem

%%% 24
\bibitem[\protect\citeauthoryear{Chebotarev and Shamis}{1998}]{ChSh98a}
\begin{barticle}
\bauthor{\bsnm{Chebotarev}, \binits{P.Y.}},
\bauthor{\bsnm{Shamis}, \binits{E.}}:
\batitle{Characterizations of scoring methods for preference aggregation}.
\bjtitle{Annals of Operations Research}
\bvolume{80},
\bfpage{299}--\blpage{332}
(\byear{1998})
\end{barticle}
\endbibitem

%%% 25
\bibitem[\protect\citeauthoryear{Gonzalez et~al.}{2019}]{GLS19a}
\begin{barticle}
\bauthor{\bsnm{Gonzalez}, \binits{S.}},
\bauthor{\bsnm{Laruelle}, \binits{A.}},
\bauthor{\bsnm{Solal}, \binits{P.}}:
\batitle{Dilemma with approval and disapproval votes}.
\bjtitle{Social Choice and Welfare}
\bvolume{53},
\bfpage{497}--\blpage{517}
(\byear{2019})
\end{barticle}
\endbibitem

%%% 26
\bibitem[\protect\citeauthoryear{{\"O}zt{\"u}rk}{2020}]{Oztu20a}
\begin{barticle}
\bauthor{\bsnm{{\"O}zt{\"u}rk}, \binits{Z.E.}}:
\batitle{Consistency of scoring rules: {A} reinvestigation of composition-consistency}.
\bjtitle{International Journal of Game Theory}
\bvolume{49},
\bfpage{801}--\blpage{831}
(\byear{2020})
\end{barticle}
\endbibitem

%%% 27
\bibitem[\protect\citeauthoryear{Brandl and Peters}{2022}]{BrPe22a}
\begin{barticle}
\bauthor{\bsnm{Brandl}, \binits{F.}},
\bauthor{\bsnm{Peters}, \binits{D.}}:
\batitle{Approval voting under dichotomous preferences: A catalogue of characterizations}.
\bjtitle{Journal of Economic Theory}
\bvolume{205},
\bfpage{105532}
(\byear{2022})
\end{barticle}
\endbibitem

%%% 28
\bibitem[\protect\citeauthoryear{Greaves and Cotton-Barratt}{2023}]{GrCB23a}
\begin{botherref}
\oauthor{\bsnm{Greaves}, \binits{H.}},
\oauthor{\bsnm{Cotton-Barratt}, \binits{O.}}:
A bargaining-theoretic approach to moral uncertainty.
Journal of Moral Philosophy
(2023)
\end{botherref}
\endbibitem

%%% 29
\bibitem[\protect\citeauthoryear{Heitzig}{2004}]{river2004mailinglist}
\begin{botherref}
\oauthor{\bsnm{Heitzig}, \binits{J.}}:
Condorcet Trees.
\url{http://lists.electorama.com/pipermail/election-methods-electorama.com/2004-October/014018.html}.
Posted on the Election Methods mailing list, October 21, 2004
(2004)
\end{botherref}
\endbibitem

%%% 30
\bibitem[\protect\citeauthoryear{Zavist and Tideman}{1989}]{zavist1989complete}
\begin{barticle}
\bauthor{\bsnm{Zavist}, \binits{T.M.}},
\bauthor{\bsnm{Tideman}, \binits{T.N.}}:
\batitle{Complete independence of clones in the ranked pairs rule}.
\bjtitle{Social Choice and Welfare}
\bvolume{6}(\bissue{2}),
\bfpage{167}--\blpage{173}
(\byear{1989})
\end{barticle}
\endbibitem

%%% 31
\bibitem[\protect\citeauthoryear{Brill and Fischer}{2012}]{brill_price_2012}
\begin{bchapter}
\bauthor{\bsnm{Brill}, \binits{M.}},
\bauthor{\bsnm{Fischer}, \binits{F.}}:
\bctitle{The price of neutrality for the ranked pairs method}.
In: \bbtitle{Proceedings of the 26th AAAI Conference on Artificial Intelligence (AAAI-12)},
pp. \bfpage{1299}--\blpage{1305}.
\bpublisher{AAAI Press},
\blocation{Palo Alto}
(\byear{2012})
\end{bchapter}
\endbibitem

%%% 32
\bibitem[\protect\citeauthoryear{Debord}{1987}]{Debo87a}
\begin{barticle}
\bauthor{\bsnm{Debord}, \binits{B.}}:
\batitle{Caract{\'e}risation des matrices des pr{\'e}f{\'e}rences nettes et m{\'e}thodes d'agr{\'e}gation associ{\'e}es}.
\bjtitle{Math{\'e}matiques et sciences humaines}
\bvolume{97},
\bfpage{5}--\blpage{17}
(\byear{1987})
\end{barticle}
\endbibitem

\end{thebibliography}

\clearpage
\appendix
\section*{Appendix}
The supplementary material consists of three parts. In \Cref{sec:app_baseballexample} we present a real life preference profile with 14 alternatives and the diagrams of River and Ranked Pairs. In \Cref{subsec:app_profiles} we give the preference profiles for the margin graphs in \Cref{fig:IPDA_otherSC} and argue why there exists a preference profile for the margin graph in \Cref{fig:example_SC_functions,fig:example_definition}. Lastly, in \Cref{subsec:app_axioms} we present the definitions and proofs omitted in \Cref{sec:Properties} for the basic axioms and for IPDA and IQDA in \Cref{subsec:app_IPDA}.

\section{Baseball Example} \label{sec:app_baseballexample}
\begin{figure*}[ht]
    \centering
    \includegraphics[width=\textwidth]{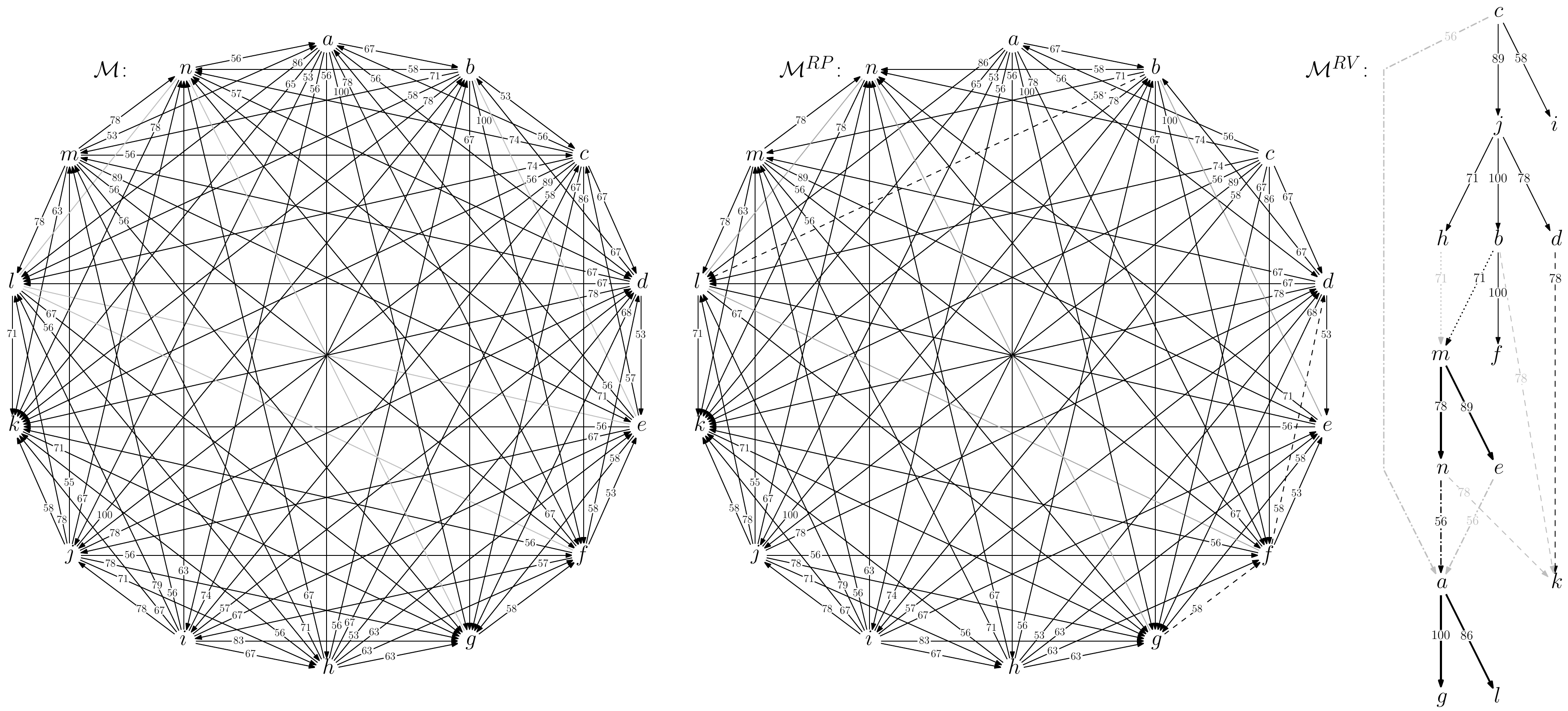}
    \caption{Large real-life example with 14 alternatives. From left to right are the margin graph, the Ranked Pairs diagram and the River diagram. Margin-zero edges are drawn in gray and all other margins are given on the edges. Lines dashed or dotted of the same kind illustrate that this set of edges would have to be decided by a tiebreaker and that only one of those edges can be in the diagram.}
    \label{fig:baseball_example}
\end{figure*}
We present a real-world example using data from a large tournament with 14 alternatives, specifically the American League baseball standings from 2004, taken from the ESPN website.\footnote{\url{https://www.espn.com/mlb/standings/grid/_/year/2004}}
In this standings grid, each cell entry $(i,j)$ represents the winning percentage of team $i$ against team $j$, calculated by dividing the number of wins by the total number of games played between the two teams.

In \Cref{fig:baseball_example}, on the left is the margin graph \margingraph\ derived from the 2004 American League baseball standings, featuring 14 alternatives labeled $A=\{a,b,\dots,m,n\}$. Gray edges indicate a margin of 0, while all other edges are labeled with their respective margins. On the middle  is the Ranked Pairs diagram $\margingraph^{\RP}$ and on the right the River diagram $\margingraph^{\RV}$ laid out as a tree. 

The Ranked Pairs diagram contains nearly every defeat from the margin graph, leading to a cluttered and complex diagram, as well as requiring numerous cycle checks during computation. 
In contrast, the River diagram keeps only a single edge among the alternatives $f$, $g$, and $k$, which are defeated 9, 12, and 13 times, respectively, which are all kept by Ranked Pairs. The River diagram is accordingly much easier to comprehend.

Dashed or dotted edges of the same style in the diagram indicate decisions that require a tiebreaker, with only one edge from each set being included in the final diagram.
Observe that the River diagram involves three such tiebreaking decisions, while the Ranked Pairs diagram has only one. However, in this particular case, the specific choice of edge within each set does not affect the final outcome.

Additionally, none of the five zero-margin edges are included in the River diagram, whereas one direction of each could potentially appear in the Ranked Pairs diagram, as seen with the edge $(n,g)$. Since these zero-margin edges do not influence the outcome, we left them gray in the diagrams.

\section{Preference Profiles for Examples} \label{subsec:app_profiles}
\smallskip
\paragraph{Preference Profile for \Cref{fig:example_SC_functions,fig:example_definition}}
It is known that for any margin graph where all margins are even (or all margins are odd), one can construct a corresponding profile of transitive preferences \citep{Debo87a}.
%More formally, this implies that for any weighted directed tournament graph with arbitrary even, integer-valued weights, we can construct a preference profile that has this tournament graph as its margin graph. 
%
%Note that we do not fix the number of voters, as it depends on the graph's weights rather than just the number of alternatives. Additionally, we want to mention that this construction not necessarily yields a minimum-size preference profile.
%\newcommand{\reverse}{R}
%
%Let $\alpha_{A'}$ denote the alphabetical order of a subset $A'\subseteq A$ of alternatives, and let $\alpha_{A'}^{\reverse}$ denote the reversed order.
%For pair of alternatives $x,y$ with margin $\margin{x,y}=2k$, we generate $k$ voters with the ranking $(x,y,\alpha_{A\setminus\{x,y\}})$ and $k$ voters with the ranking $(\alpha_{A\setminus\{x,y\}}^{\reverse},x,y)$.
%
%These $2k$ preferences yield $\margin{x,y}=2k$ and $\margin{z,z'}=0$ for all $z,z'\in A$ with $\{z,z'\}\neq\{x,y\}$.
%By repeating this process for each pair, we enforce the specified, even margin for this pair, without affecting the margin of the other pairs.
%This approach produces a preference profile with $\sum_{x\in A, y\in A\setminus\{x\}}\margin{x,y}$ voters.

\paragraph{Preference Profiles for \Cref{fig:IPDA_otherSC}}
% Note that this is not necessarily the lowest number of voters to achieve this margin graph. Furthermore, this does not directly work for margin graphs with Pareto-dominated alternatives.
The above-mentioned process for constructing profiles cannot be used for the margin graphs that demonstrate the violation of independence of Pareto-dominated alternatives for Split Cycle and its refinements apart from River.
For such margin graphs, alternative $b$ must Pareto-dominate alternative $a$, i.e., the margin of $b$ over $a$ must equal the number of voters. Since the above-mentioned construction process involves changing the number of voters during construction, it cannot be applied here.
Therefore, we provide specific preference profiles for the presented margin graphs.

\medskip

\noindent \textbf{Preference Profile $\Pref_1$:}
\[
\begin{array}{llllllll}
7 & 5 & 4 & 2 & 1 & 1 & 1 & 1 \\
\hline & \\[-2ex]
c & b & d & b & d & c & b & d \\
b & a & c & a & b & b & c & b \\
a & d & b & c & a & d & a & c \\
d & c & a & d & c & a & d & a 
\end{array}
\]
\noindent \textbf{Preference Profile $\Pref_2$:}
\[
\begin{array}{lllllllll}
7 & 5 & 2 & 2 & 2 & 1 & 1 & 1 & 1 \\
\hline & \\[-2ex]
d & b & c & b & c & c & b & d & b \\
b & a & d & a & d & b & d & c & c \\
a & c & b & d & b & d & a & b & a \\
c & d & a & c & a & a & c & a & d 
\end{array}
\]
\noindent \textbf{Preference Profile $\Pref_3$:}
\[
\begin{array}{clllllllllllll}
12 & 8 & 8 & 7 & 6 & 5 & 4 & 3 & 3 & 2 & 2 & 2 & 1 & 1 \\
\hline & \\[-2ex]
e  & c & d & d & e & e & b & c & b & b & d & d & d & c \\
b  & e & a & c & d & c & a & d & d & c & b & c & e & b \\
a  & b & b & e & b & d & c & a & a & a & a & b & a & a \\
c  & a & c & b & a & b & d & b & c & d & c & e & b & d \\
d  & d & e & a & c & a & e & e & e & e & e & a & c & e 
\end{array}
\]
\noindent \textbf{Preference Profile $\Pref_4$:}
\[
\begin{array}{cllllllllllll}
8 & 6 & 5 & 4 & 3 & 2 & 2 & 2 & 2 & 1 & 1 \\
\hline & \\[-2ex]
 d & b & e & b & d & b & b & d & d & d & c \\
 a & a & c & a & a & d & c & b & c & e & b \\
 b & e & d & c & b & a & a & a & b & a & a \\
 c & d & b & d & c & c & d & c & e & b & d \\
 e & c & a & e & e & e & e & e & a & c & e 
\end{array}
\]

\renewcommand{\restatemarker}{}
\section{Definitions and Proofs for Basic Axioms} \label{subsec:app_axioms}

    \subsection{Anonymity and Neutrality}
    %%%%%%%%%%%%%%%%%%%%%%% Anonymity and neutrality %%%%%%%%%%%%%%%%%%%%%%%
    \begin{definition} [Anonymity and Neutrality] \label{def:AnonymityandNeutrality} A social choice function \SCF satisfies
       \begin{enumerate}
           \item \defstyle{anonymity} if for any permutation $\pi\colon N \rightarrow N$ of voters and all $A, \Pref, \Pref'$ for which $\succ'_i = \succ_{\pi(i)}$ holds for all $i\in N$, we have $\SCF(\Pref,A) = \SCF(\Pref',A)$;
           \item \defstyle{neutrality} if for any bijection $\pi\colon A\rightarrow B$ of alternatives and all $A,\Pref,\Pref'$ for which $(xP_iy) \leftrightarrow (\pi(x)\Pref'_i\pi(y))$ holds for all $i\in N$ and $x,y\in A$, we have $\pi(\SCF(\Pref,A)) = \SCF(\Pref',\pi(A))$.
       \end{enumerate}
    \end{definition}
    For River, both anonymity and neutrality are dependent on the tiebreaker in the same way as for Ranked Pairs: For uniquely-weighted profiles, River is anonymous and neutral, while for general profiles it depends on the equipped tiebreaker. Note that each deterministic tiebreaker must either violate anonymity or neutrality.
    \begin{proposition} \label{prop:AnonymityNeutrality}
        River satisfies anonymity and neutrality in uniquely-weighted profiles.
    \end{proposition}
    The reasoning for River is completely analogous to the arguments for Ranked Pairs \cite{tideman1987independence}.
    
    One example of an anonymous and neutral but non-deterministic tiebreaker is ordering edges with the same margin uniformly at random.
    An example of a deterministic tiebreaker that is neutral and Pareto-consistent but not anonymous or quasi-Pareto-consistent is to order edges lexicographically according to the ranking of the first voter.
    However, if we choose the voter, whose ranking we take as the base of the lexicographic ordering, at random, this yields a deterministic, neutral, and Pareto-consistent tiebreaker.
    
    An example of a deterministic tiebreaker that is anonymous and consistent but not neutral or Pareto-consistent is to order edges lexicographically according to a pre-specified fixed ordering on the universe of all possible alternatives. 
    An example of a deterministic tiebreaker that is quasi-Pareto-consistent and neutral but not anonymous is to first determine the set $Q$ of quasi-Pareto domination edges and the set $\bar Q$ of all other edges, then order both of them internally using any deterministic and neutral tiebreaker, and finally put $Q$ before $\bar Q$.
%    \markus{I wouldn't employ a random tiebreaker here, as we are talking about deterministic SCFs in this paper. That means that neutrality and anonymity cannot be achived together. Probably it is also discussed by \cite{tideman1987independence} how to get anonymity (without neutrality) and neutrality (without anonymity)? Something like ``use a fixed order'' or ``use the preferences of the first voter''...? }

    The existence of an anonymous and neutral tiebreaker implies the following:
    \begin{proposition}
         River satisfies anonymity and neutrality in general profiles when equipped with an appropriate tiebreaker. 
    \end{proposition}

    % \subsubsection{Reversal Symmetry}
    % %%%%%%%%%%%%%%%%%%%%%%% Reversal Symmetry %%%%%%%%%%%%%%%%%%%%%%%
    % \begin{definition}[Reversal Symmetry] \label{def:ReversalSymmetry}
    %     A social choice function \SCF satisfies \defstyle{reversal symmetry} if for any \Pref with $\SCF(\Pref)=\{x\}$, $x\notin SCF(\Pref^r)$, where $\Pref^r$ is obtained from $\Pref$ by reversing every ballot.
    % \end{definition}
    % \begin{proposition} \label{prop:ReversalSymmetry}
    %     River satisfies reversal symmetry.
    % \end{proposition}
    % \begin{proof}
    % I don't think that Ranked Pairs does not satisfy this, nor does River.
    % % Let River be equipped with an arbitrary tiebreaker and let $\Pref$ be a preference profile such that $\RV(\Pref)=\{x\}$. Then since \RV is a refinement of \SC, $x\in \SC(\Pref)$.
    % % Now, as Split Cycle satisfies reversal symmetry (\cite{holliday2020split}), $x\notin\SC(\Pref^r)$ and also $x\notin\SC(\Pref^r)$.
    
    % % Then $x$ has at least one outgoing edge $(x,y)\in E(\margingraphRVExt{\Pref})$. We claim that $(y,x)\in E(\margingraphRVExt{\Pref^r})$.
    % % Reversing every preference relation leads to reversing the margin of every edge, \ie $\marginExt{x,y}{\Pref}=\marginExt{y,x}{\Pref^r}$. Therefore, the ordering of the edges by increasing margin stays the same, while the direction of the edges are reversed.
    % % Following the River procedure for $\Pref^r$, every edge is added/disregarded just as for the original preferences $\Pref$ and therefore $(y,x)\in E(\margingraphRVExt{\Pref^r})$. Thus, $x\notin\RV(\Pref^r)$.
    % \end{proof} 

    \subsection{Monotonicity}
    %%%%%%%%%%%%%%%%%%%%%%% Monotonicity %%%%%%%%%%%%%%%%%%%%%%%
    \begin{definition} [Monotonicity] \label{def:Monotonicity}
        A social choice function \SCF is \defstyle{monotonic} if for any \Pref and $x\in\SCF(\Pref)$, we have $x\in\SCF(\Pref^{x+})$ where $\Pref^{x+}$ is the preference profile obtained from $\Pref$ by putting the alternative $x$ one position higher in any voters' ballot.
    \end{definition}
    \begin{proposition} \label{prop:Monotonicity}
        River satisfies monotonicity.
    \end{proposition}
    \begin{proof}
        Let $x\in\RV(\Pref)$ and $\Pref_{x+}$ be the preference profile obtained from $\Pref$ by putting $x$ one position higher in any voters' ballot. Let $E$ and $E_+$ denote the edges of the margin graphs of $\Pref$ and $\Pref_{x+}$, respectively, and $E^{\RV}$ and $E^{\RV}_+$ the edges of the corresponding River diagrams.
        
        There is an alternative $y$ with whom $x$ switched position in that voters ballot. Therefore, $\marginExt{x,y}{\Pref_{x+}} = \marginExt{x,y}{\Pref}+2$ is the only change in the margin graph.
        
        If $\marginExt{x,y}{\Pref}<0$, $(y,x)\in E$, but since $x\in\RV(\Pref)$, $(y,x)\notin E^{\RV}$. Therefore, it was the lowest margin of a majority cycle in $E^{\RV}$. As $(y,x)$ is the only edge that is different in $E_+$, that majority cycle is also contained in $E^{\RV}_+$ and $(y,x)\notin E^{\RV}_+$ as it is the lowest margin edge of a majority cycle. 
        
        If $\marginExt{x,y}{\Pref}>0$, $(x,y)\in E$. Now, if there was a relevant cycle containing $(x,y)$ in \Pref, then that is still valid in $\Pref_+$.
        If $(x,y)\notin E^{\RV}$, but $(x,y)\in E_+^{\RV}$, then there is an edge $(w,y)\in E^{\RV}$, but $(w,y)\notin E^{\RV}_+$.
        Now, for any cycle in $E^{\RV}$ eliminating an incoming edge $(z,x)$ of $x$ using the now missing $(w,y)$, we have $\margin{x,y}\geq\margin{w,y}$. So, $(z,x)$ is still the lowest margin edge and not added to $E_+^{\RV}$.
    \end{proof}
    \subsection{Condorcet Winner and Loser}
    %%%%%%%%%%%%%%%%%%%%%%% condorcet %%%%%%%%%%%%%%%%%%%%%%%
    \begin{definition} \label{def:Condorcet}
        A social choice function \SCF satisfies the \defstyle{Condorcet winner} (resp. \defstyle{loser}) \defstyle{criterion} if for any preference profile $\Pref$ and $x\in A$, if $x$ is the Condorcet winner (resp. loser) then $\SCF(\Pref)=\{x\}$ (resp $x\notin \SCF(\Pref)$). If \SCF satisfies the Condorcet winner criterion, we say that \SCF is \defstyle{Condorcet consistent}.
    \end{definition}
    \begin{proposition} \label{prop:Condorcet}
        River satisfies the Condorcet winner and Condorcet loser criterion.
    \end{proposition}
    \begin{proof}
    Let $x\in A$ be a Condorcet winner. Then for every $y\in A\setminus\{x\}$, we have $\margin{x,y}>0$. It follows that $x$ does not have any incoming edges in the margin graph, hence $x\in \RV(\Pref)$.
    Let now $x$ be a Condorcet loser. Then $\margin{x,y}<0$ and $x$ has only incoming edges in the margin graph. As $x$ is not contained in any cycles, one such edge will be in $E(\margingraphRV)$. Hence, $x\notin \RV(\Pref)$.
    \end{proof}
    Observe that this criterion also follows from River being a refinement of Split Cycle.
    
    %%%%%%%%%%%%%%%%%%%%%%% Smith %%%%%%%%%%%%%%%%%%%%%%%
    % \begin{definition}[Smith set]
    %     Let $\Pref$ be a preference profile.
    %     The \defstyle{Smith} or \defstyle{\GETCHA set} is defined as
    %     \[\GETCHA (\Pref) = \bigcap\{B\subseteq A \colon b\succ x \text{ for all }b\in B, x\in A\setminus B\}.\]
    % \end{definition}
    % \begin{definition}[Smith criterion] \label{def:Smith}
    %     A social choice function \SCF satisfies the \defstyle{Smith criterion} if for any preference profile $\Pref$, holds $\SCF(\Pref)\subseteq \GETCHA(\Pref)$. 
    % \end{definition} \noindent
    % River satisfies the Smith criterion, which simply follows from River being a refinement of Split Cycle and the fact that Split Cycle satisfies the Smith criterion: $\RV\subseteq\SC\subseteq \GETCHA$.
    % \begin{corollary}\label{prop:Smith}
    %     River satisfies the Smith criterion.
    % \end{corollary}

    \subsection{Pareto Efficiency}
    %%%%%%%%%%%%%%%%%%%%%%% Pareto %%%%%%%%%%%%%%%%%%%%%%%
    \begin{definition} [Pareto efficiency] \label{def:Pareto}
        A social choice function \SCF is \defstyle{Pareto-efficient} if for any preference profile $\Pref$ and $a,b\in A_\Pref$, if all voters in $N_\Pref$ rank $a$ above $b$, then $b\notin\SCF(\Pref)$.
        %\SCF satisfies \defstyle{Pareto} if for any $\T$ and $x,y\in V(\T)$, $\margin{y,x}=0$ implies $y\notin S(\T)$.
    \end{definition}
    \begin{proposition} \label{prop:Pareto}
        River is Pareto-efficient.
    \end{proposition}
    \begin{proof}
    If every voter ranks $x$ above $y$, we have $\margin{x,y}=n$ and this edge has the highest possible margin. Because all individual preference relations are acyclic, the graph of all such Pareto dominations is acyclic as well. Therefore, when $(x,y)$ is considered for addition to the river diagram, no cycle from $y$ back to $x$ can have been added already. So if $(x,y)$ is not added, some other edge $(z,y)$ must have been added already. In either case, $y\notin \RV(\Pref)$.
    \end{proof}

\section{Proofs for IPDA and IQDA} \label{subsec:app_IPDA}
We present the proofs for \Cref{thm:IPDA_River_TB} and \Cref{thm:IQDA_River} omitted in the paper.

\thmIPDARiverTB*
\begin{proof}[Proof of \Cref{thm:IPDA_River_TB}]
This is literally the same as the proof of \Cref{thm:IQDA_River}, only with all occurrences of the prefix ``quasi-'' removed.
\end{proof}

    For a non-\parecon tiebreaker, one can generate a preference profile for which River with that tiebreaker violates IPDA, as shown in the following example.
    \begin{figure}[ht]
    \centering
    \[
    \begin{array}{clllllllllll}
    12 & 9 & 9 & 9 & 7 & 6 & 5 & 5 & 4 & 2 & 1 & 1 \\
    \hline & \\[-2ex]
    y  & z & b & a & y & z & b & z & z & b & y & y \\
    x  & y & a & z & x & b & a & y & a & z & a & b \\
    b  & x & z & b & a & a & y & x & y & y & x & x \\
    a  & a & y & y & z & y & x & b & x & x & z & a \\
    z  & b & x & x & b & x & z & a & b & a & b & z
    \end{array}
    \]
    \includegraphics[width=0.5\textwidth]{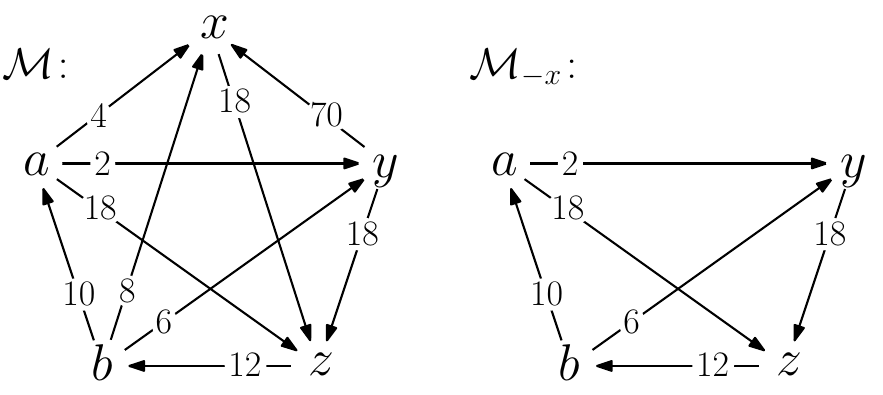}
    \caption{Preference profile and margin graph for  \Cref{ex:River_IPDA_nonparecon}, where a non-\parecon tiebreaker can lead to River violating IPDA.}
    \label{fig:example_IPDA_TB_counterexample}
\end{figure}
    \begin{example} \label{ex:River_IPDA_nonparecon}
        Consider River equipped with a non-\parecon tiebreaker.
        In the profile shown in \Cref{fig:example_IPDA_TB_counterexample} with 5 alternatives and 68 voters, $y$ Pareto-dominates $x$, and $\margin{y,z} = \margin{x,z} = \margin{a,z}=18$. 
        
        In \Pref, the edge $(y,x)$ is processed first and added to the River diagram. Depending on the tiebreaker, one of $(y,z)$, $(x,z)$ or $(a,z)$ is processed and added next, while the others, which would create branchings, are not added.
        This is also the first step of River in $\Pref_{-x}$.
        
        Now, assume the non-consistent tiebreaker orders $(a,z)$ first for $\Pref$, but $(y,z)$ first for $\Pref_{-x}$. We show that this results in $\RV(\Pref_{-x})\neq \RV(\Pref)$, \ie River violating IPDA.
        
        In \Pref, $(z,b)$ is added next, followed by the rejection of $(b,a)$ due to a majority cycle.
        The edges $(b,x)$ and $(b,y)$ are processed next, with $(b,x)$ being excluded due to a branching, and $(b,y)$ being added.
        The final edges $(a,x)$ and $(a,y)$ are both rejected due to branchings, leading to $\RV(\Pref) = {a}$.

        In $\Pref_{-x}$, $(z,b)$ and $(b,a)$ are both added, as $(a,z)$ is not in the diagram. However, both $(b,y)$ and $(a,y)$ create cycles and are rejected, resulting in $\RV(\Pref_{-x}) = {y}$.
    \end{example}

Now we turn to the more permissive notion of quasi-Pareto domination and the proof that River satisfies the corresponding notion of independence IQDA. First, let us define quasi-Pareto-consistent tiebreakers.
\begin{definition}
            A \parecon tiebreaker is called \defstyle{quasi-Pareto-consistent} if for any two edges $(y,x)$, $(y',x)$ with equal margins, of which $y$ but not $y'$ quasi-Pareto-dominate $x$, the tiebreaker ranks $(y,x)$ before $(y',x)$, and whenever $y$ quasi-Pareto-dominates $x$, any edge $(x,z)$ is always ranked after both the edges $(y,z)$ and $(y,x)$.
        \end{definition}
%    Note that quasi-Pareto-consistent tiebreakers exist since quasi-Pareto domination is acyclic. One can thus pick any linear ordering $\prec$ on $A$ that has $y\prec x$ whenever $y$ quasi-Pareto-dominates $x$, and let the tiebreaker rank $(y,x)$ before $(z,w)$ if and only if either (i) $(x,y)$ is a quasi-Pareto domination but $(z,w)$ is not, or (ii) both $(x,y)$ and $(z,w)$ are quasi-Pareto dominations and $(x,y)$ is lexicographically before $(z,w)$ according to $\prec$, or (iii) both $(x,y)$ and $(z,w)$ are not quasi-Pareto dominations and $(x,y)$ is lexicographically before $(z,w)$ according to $\prec$. \jobst{this might not give us a *consistent* tiebreaker!}
% \markus{Definition of quasi Pareto domination is missing!}
    
% \thmIQDARiver*
\thmIQDARiver*
\begin{proof}[Proof of \Cref{thm:IQDA_River}]
        Let \Pref be a preference profile and $x\in A_\Pref$ an alternative that is quasi-Pareto-dominated. From those $y$ that quasi-Pareto-dominate $x$, choose that one whose edges are ranked first by the tiebreaker. %\markus{be more precise. the tiebreaker ranks edges (or unordered pairs...), but not alternatives.}
        Let $E=E(\margingraphRVExt{\Pref})$, $\EnoX=E(\margingraphRVExt{\Pref_{-x}})$ be the edge set of the River diagrams arising with and without $x$, respectively. 
        
        By choice of $y$, $(y,x)$ is processed before any other edge $(z,x)$, hence it cannot form a branching. Because either no ties exist or the tiebreaker is quasi-Pareto-consistent, no edge $(x,z)$ has been added at this point. Hence $(y,x)$ can also not form a cycle with earlier added edges. Hence the edge is added to $E$.
        
        Next, we show that also $(x,z)\notin E$ for all $z\in A\setminus\{x\}$.
        Assume towards contradiction that $(x,z)$ is added to $E$ when processed. Then for all $z'\in A\setminus\{x\}$ we cannot have added $(z',z)$ earlier, in particular not $(y,z)$.
        Since $y$ quasi-Pareto-dominates $x$ and either no ties exist or the tiebreaker is quasi-Pareto-consistent, $(y,z)$ is processed before $(x,z)$. Therefore, the same argument as in \Cref{thm:IPDA_River_uw} holds and $(x,z)\notin E$.

        Finally, we show that for all $z\neq x\neq z'$, we have $(z,z')\in E$ if and only if $(z,z')\in\EnoX$. Since either no ties exist or the tiebreaker is consistent, the order in which these edges are processed is the same for \Pref and for $\Pref_{-x}$. Therefore, the argument of \Cref{thm:IPDA_River_uw} holds here too.
        
        Therefore, all edges in $E$ apart from $(y,x)$ are in $\EnoX$ and $\RV(\Pref)=\RV(\Pref_{-x})$.
    \end{proof}
    



\end{document}